%% file: cat-state-preparation.tex
\documentclass[a4paper,onecolumn,superscriptaddress,11pt%
              ,allowtoday%
              ]{quantumarticle}
\pdfoutput=1

\input{cnotsyn.sty}
\usepackage{graphicx}
\usepackage{multirow}
\usepackage{booktabs}
\usepackage{hyperref}
\usepackage{rotating}

\input{externalize-disable.tex}

\title{SpiderCat: Optimal Fault-Tolerant Cat State \newline Preparation}

\author{Andrey Boris Khesin}
\affiliation{University of Oxford, Oxford, United Kingdom}

\author{Sarah Meng Li}
\affiliation{University of Amsterdam, Amsterdam, Netherlands}
\affiliation{QuSoft, Amsterdam, Netherlands}

\author{Boldizsár Poór}
\affiliation{University of Oxford, Oxford, United Kingdom}

\author{Benjamin Rodatz}
\affiliation{University of Oxford, Oxford, United Kingdom}

\author{John van de Wetering}
\affiliation{University of Amsterdam, Amsterdam, Netherlands}
\affiliation{QuSoft, Amsterdam, Netherlands}

\author{Richie Yeung}
\affiliation{University of Oxford, Oxford, United Kingdom}
\date{}
\begin{document}

\maketitle


\begin{abstract}
	The ability to fault-tolerantly prepare $\textsf{CAT}\xspace$ states, also known as multi-qubit GHZ states, is an important primitive for quantum error correction. It is required for Shor-style syndrome extraction, and can also be used as a subroutine for doing fault-tolerant state preparation of CSS codewords. Existing approaches to fault-tolerant $\textsf{CAT}\xspace$ state preparations have been found using computationally expensive heuristics involving SAT solving, reinforcement learning, or exhaustive analysis.
	In this paper, we constructively find optimal circuits for $\textsf{CAT}\xspace$ states in a more scalable way. In particular, we derive formal lower bounds on the number of CNOT gates required for circuits implementing $n$-qubit $\textsf{CAT}\xspace$ states that do not spread errors of weight at most $t$ for $1\leq t \leq 5$. We do this by using fault-equivalent rewrites of ZX-diagrams to reduce it to a problem of characterising certain 3-regular simple graphs. We then provide families of such optimal graphs for infinitely many values of $n$ and $t\leq5$.
    By encoding the construction of optimal graphs as a constraint satisfaction problem we find explicit constructions for circuits that match this lower bound on CNOT count for all $n\leq50$ and $t \leq 5$ and for nearly all pairs $(n,t)$ with $n\leq 100$ and $t\leq 5$ or $n\leq 50$ and $t\leq 7$, significantly extending the regimes that were achievable by previous methods and improving the resource counts for existing constructions. We additionally show how to trade CNOT count against depth, allowing us to construct constant-depth fault-tolerant implementations using $O(n)$ ancilla and $O(n)$ CNOT gates.
\end{abstract}




\section{Introduction}\label{sec:intro}

As larger quantum computing devices become available, the first experimental demonstrations of fault-tolerant quantum computation have begun to emerge~\cite{google2023suppressing,google2025quantum,abobeih2022fault,sales2025experimental,zhang2025demonstrating}. These results show that, by carefully structuring unitaries and measurements, it is possible to suppress the propagation of physical errors, detect faults as they occur, and actively correct them. While in the classical setting error correction is typically confined to communication and storage, while logical operations are assumed to be essentially noise-free, in quantum computing, however, every gate, measurement, and even idle qubit is inherently noisy~\cite{errormitigation,Ferracin2024efficiently,gottesman2009introduction}. As a result, fault tolerance must be treated as a global design constraint: the entire circuit architecture must be engineered so that errors remain controlled throughout the computation~\cite{gottesman1997stabilizer,gottesmanbook,preskill1998fault,shor1996fault}.


While the foundational feasibility results for fault-tolerant quantum computing have been established for decades~\cite{biasednoise,gottesman2013fault,shorsyndrome}, it is only in recent years that fault-tolerant compilation has emerged as an active area of research, with protocols that significantly optimise individual components of fault-tolerant circuits~\cite{Chamberland2018flagfaulttolerant,Chamberland2019faulttolerantmagic,chamberland2020very,chao2018fault,forlivesi2025flag,rodatz2024floquetifying,rodatz2025fault}. A particular component that has received a lot of attention is the fault-tolerant preparation of stabiliser states for specific (small) codes. Existing approaches include reinforcement learning~\cite{Porotti2022deepreinforcement,rl}, constraint-satisfaction solvers~\cite{automatedsyn}, brute-force search~\cite{forlivesi2025flag}, and systematic analyses of error propagation and the corresponding verification circuits~\cite{gotoMinimizingResource2016,kanomata2025fault,golay}, all aimed at constructing fault-tolerant state-preparation circuits for specific, typically small, codes.

Although these methods are effective for circuits of up to roughly 30–70 qubits (depending on the approach) and under the assumption of low-weight errors, they do not readily scale, as they rely on analysing all possible errors across the entire circuit, either directly or through carefully tailored problem reductions. Here, we present a fundamentally different approach to stabiliser-state preparation, with a focus on the fault-tolerant preparation of \emph{\CAT states} (i.e.~generalised GHZ states) of the form $\left(\ket{0\cdots 0}+\ket{1\cdots 1}\right)/\sqrt{2}$. These states are a key primitive for constructing general stabiliser states~\cite{forlivesi2025flag}, implementing Shor-style syndrome extraction~\cite{bacon_shor_syndrome,shorsyndrome,Tansuwannont2023adaptivesyndrome} (including optimised variants~\cite{rodatz2025fault}), realising fault-tolerant logical operations~\cite{near_term_code_switching,ldpc_encoder}, enabling magic-state distillation~\cite{chamberland2020very,gotoMinimizingResource2016}, and serving as an important benchmark for quantum hardware~\cite{zimboras2025eu}.

\begin{table}[ht]
  \centering
  \begin{tabular}{llllll}
    \toprule
    \textbf{Method} & \textbf{CNOT count} & \textbf{Depth} & \textbf{Ancillas}& \textbf{max $t$}& \textbf{max $n$} \\
    \midrule
    FA~\cite{forlivesi2025flag} & $O(n)$ & $O(n)$ & $O(n)$ & $6$ & $24$\\
    PWW~\cite{peham2026optimizing}  & $\leq 3n$ & $O(\log n)$ & $n$ & $9$ & $19$\\
    CB~\cite{Chamberland2018flagfaulttolerant} & $O(n)$ & $O(n)$ & $O(n)$ & $5$ & $8$\\
    \midrule
    \textbf{\scriptsize Recursive}  & $n\log (t+1)+n$ & $2\log (t+1)+2$ & $\boldsymbol{\leq \frac12 n}$ & $\boldsymbol\infty$ & $\boldsymbol\infty$\\
    \textbf{\begin{tabular}[l]{@{}l@{}}\scriptsize Optimal\\ \scriptsize \quad\ \ (deep $||\,$ SAT)\end{tabular}}  & $\boldsymbol{(r_t+1)n+1}$ & $\leq (r_t+1)n+1$ & $\boldsymbol{\leq\frac{r_t}2n+1}$ & $5 \,||\, 7$ & $\boldsymbol\infty \,||\, 50$\\
    \textbf{\scriptsize Optimal (shallow)}  & $\leq \left(\frac{29r_t+26}{10}\right)n$ & $\boldsymbol{3}$ & $\left(\frac{12r_t+8}5\right)n$ & $5$ & $\boldsymbol\infty$\\
    \textbf{\scriptsize Ramanujan (some $n$)}  & $\leq 37.4n$ & $\boldsymbol{3}$ & $\leq30.4n$ & $\boldsymbol\infty$ & $\boldsymbol\infty$\\
    \textbf{\scriptsize Ramanujan (all $n$)}  & $\leq5\times10^6n$ & $\boldsymbol{3}$ & $\leq4\times10^6n$ & $\boldsymbol\infty$ & $\boldsymbol\infty$\\
    \hline
\end{tabular}
\caption{Comparison of different $n$-qubit \CAT-state preparation methods that are fault-tolerant up to faults of weight $t$. Depth is the CNOT depth, and Ancillas is the maximum number of ancilla in use at one time, with reuse allowed. The number $r_t$ is the optimal \emph{vertex ratio} given in \cref{def:vertex-ratio}. The max value of $n$ is reported for the largest $t$. In the row labelled `Optimal (deep $||$ SAT)'', we report results obtained from the graph-theoretic constructions of \cref{sec:ft-cat-graph}, as well as circuits found using constraint satisfaction as described in \cref{sec:cat-state-generation}.}
  \label{tab:asymptotics}
\end{table}


We use the framework of \emph{fault-tolerance by construction}~\cite{rodatz2025fault}, where we represent a circuit by a \emph{ZX-diagram}~\cite{coeckeInteractingQuantumObservables2008}, and then apply carefully chosen local rewrites to this diagram that preserve \emph{fault-equivalence}. In this way we get an optimised circuit that has the desired level of fault-tolerance by design. We prove a number of results using this framework.

First, we find a simple intuitive construction of $t$-fault-tolerant $n$-qubit \CAT states which requires no expensive computation to construct, while achieving gate count $O(n)$, qubit count $O(n)$, and circuit depth $O(\log t)$ (this method is labelled as `recursive' in \cref{tab:asymptotics}). This is more scalable than the \CAT-state preparation circuits of~\cite{peham2026optimizing} which have depth $O(\log n)$. Moreover, our method does not require the expensive solving of large SAT instances which restrict that method to $n\leq 47$ for $t= 5$ (or $n\leq 19$ for $t= 9$).

Then, we show that any fault-tolerant \CAT state construction can be viewed as a particular type of 3-regular graph which is `resistant' to being disconnected using a small number of edge cuts. This allows us to use graph-theoretic techniques to derive for the first time concrete non-trivial lower bounds on the required CNOT count. We obtain optimal bounds for fault weight $t\leq 5$, and a (non-optimal) bound for arbitrary $t$ using Ramanujan graph constructions.



Finally, we show how to construct graphs that meet these lower bounds in practice, producing CNOT circuits for $t$-fault-tolerant $n$-qubit \CAT states ($n \leq 100$, $t \leq 5$) with provably optimal CNOT count.
The scripts constructing these \CAT-state preparation circuits, together with the circuits, the benchmarking and visualisation scripts, and simulation data can be found in the project’s GitHub repository~\cite{cat2026}.

\cref{tab:asymptotics} summarises the improvements of our \CAT state constructions by comparing the proposed strategies with previous state-of-the-art methods~\cite{Chamberland2018flagfaulttolerant,forlivesi2025flag,peham2026optimizing} across the metrics of CNOT count, circuit depth, ancilla usage, and achievable level of fault-tolerance. It shows that for each parameter regime---minimising CNOT count, depth, or ancillae---there is a corresponding optimisation algorithm with provable guarantees. The recursive, optimal (deep/SAT), and optimal (shallow) constructions together span the trade-off landscape, allowing resources to be tuned to the target constraint while remaining asymptotically optimal. Overall, our work provides a practical toolbox that turns these trade-offs into explicit design choices.

The remainder of this paper is organised as follows. \cref{sec:preliminaries} introduces the ZX-calculus and the notion of fault equivalence used throughout the paper. \cref{sec:recursive-cat} presents a scalable and low-depth way of constructing fault-tolerant \CAT states recursively. \cref{sec:circ-extract} shows how \CAT state constructions can be characterised using 3-regular marked graphs. \cref{sec:lower-bounds} develops graph-theoretic lower bounds on CNOT count and proves optimality results. \cref{sec:cat-state-generation} gives explicit constructions and algorithms that achieve these bounds in practice, including SAT- and graph-based search methods. \cref{sec:simulation-results} then presents benchmarking results of our method against others, comparing resource overhead and simulated performance. \cref{sec:conclusion} concludes with a discussion on implications and future directions. For readability, most technical proofs are deferred to the appendix.

\section{Preliminaries}\label{sec:preliminaries}

The ZX-calculus~\cite{coeckeInteractingQuantumObservables2008,coeckePicturingQuantum2017,KissingerWetering2024Book}  is a diagrammatic language for reasoning about quantum computations. 
A ZX-diagram represents linear map as a graph of \emph{spiders}, corresponding to a particular type of tensor network, which can represent any linear map between qubits. 
In this work, however, we will restrict attention to the Pauli fragment of the ZX calculus, which is universal for circuits consisting of Pauli gates, CNOTs, state preparations and measurements~\cite{kissinger2022phase,KissingerWetering2024Book}. 
Beyond representing quantum circuits, the ZX-calculus also allows reasoning about them; it comes equipped with a set of rewrites that change the diagram while preserving the underlying linear map it represents, this allows the transformation of quantum circuits while preserving their underlying computation.

\begin{definition}
  We define a \emph{Pauli ZX-diagram} $D$ as a graph $G = (V, E)$ with phases $\alpha \colon V \to \{0, \pi\}$ and vertex types $t \colon V \to \{Z, X, In\text{(put)}, Out\text{(put)}\}$.
  \emph{Boundary} vertices, i.e.\@ input or output vertices, are restricted to be of degree one. 
  We call the vertices of type Z or X \emph{spiders} and the edges in a ZX-diagram are often referred to as \emph{legs}.
  For a ZX-diagram $D$, we denote the linear map it represents, its \emph{interpretation}, by $\intf{D}$.
  When an edge is shared by a boundary node, it is called a \emph{free edge}. 
	An \emph{internal spider} is a spider that has no free edge. 
	A \emph{boundary spider} is a spider that has at least one free edge.
\end{definition}

The interpretation function $\intf{\cdot}$ associates a particular linear map to each spider, which gets composed together by tensor products and linear map compositions so that each diagram represents some linear map between qubits. For the details see \autoref{appendix:zx-intro}.
As ZX-diagrams represent linear maps, they can express quantum circuits.
We present a translation table between common quantum gates and corresponding ZX-diagrams in \autoref{fig:mappings}. Note here that the $n$-qubit \CAT state $\ket{\CAT^n} := \ket{0}^{\otimes n} + \ket{1}^{\otimes n}$ (the scalar is ignored), a generalisation of GHZ state to $n$ qubits, corresponds to just a single $Z$-spider with $n$ legs.

\begin{figure}[!tbh]
	\begin{equation*}
		\scalebox{0.9}{\input{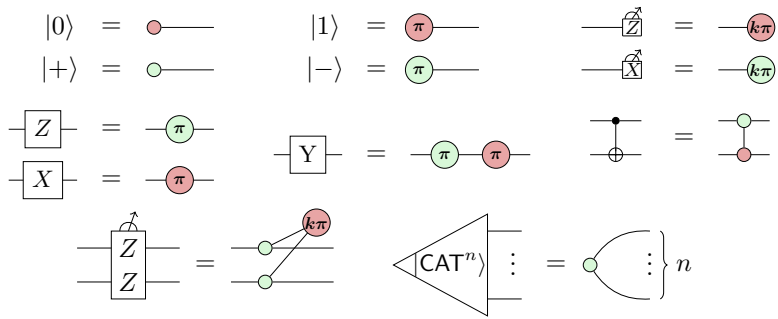}}
	\end{equation*}
  \vspace{-0.5cm}
  \caption{Mapping basic quantum states, Pauli operators, and Pauli measurements to Pauli ZX-diagrams, $n\in \N$ and $k \in \{0,1\}$. As we note in the next section, the direction of wires is irrelevant, so we can write CNOT gates using vertical wires between spiders without ambiguity. These equalities hold up to some known scalar factor, which is not relevant in this paper.}
  \label{fig:mappings}
  \vspace{-.3cm}
\end{figure}

\begin{remark}\label{rem:post-selection}
	In \cref{fig:mappings}, we write the $Z$-basis, $X$-basis, and $ZZ$-measurements using a variable $k\in\{0,1\}$ which denotes the outcome of the measurement. In this paper, we will only consider protocols for preparing \CAT states that are post-selected on the $0$ outcome, since our measurements will flag the presence of errors and we will abort the protocol if a different measurement outcome is observed. Hence, measurements in this paper will be represented by the $k=0$ branch and correspond to diagrams with all phases equal to $0$.
\end{remark}

ZX-diagrams are equipped with a set of rewrite rules known as the \emph{ZX-calculus}, which enables graphical reasoning about quantum circuits. Specifically, in a ZX-diagram, \emph{only connectivity matters} (\textsc{OCM}): their interpretation depends only on how spiders are connected, not on their placement on the page. Hence, diagrams that differ only by rearranging spiders represent the same linear map. The ZX-calculus provides rewrites that transform diagrams while preserving their interpretations. Some of these rules are shown in \autoref{fig:pauli-zx} of \cref{appendix:zx-intro}. In this work, however, we will require a refined set of rewrite rules.


In~\cite{rodatz2025fault}, Rodatz et al. introduce a finer-grained notion of equivalence for analysing how circuits behave under noise.
Their approach was motivated by the observation that faults may propagate in different ways across circuits that nevertheless implement the same operation in the fault-free case. 
When comparing and transforming circuits, it is therefore necessary to consider the effect of faults on circuits. 
We first define what we mean by faults on a ZX-diagram.
Let $\Pauli_n$ denote the group of $n$-qubit Pauli operators. 

\begin{definition}[Faults locations, faults]\label{def:fault-location}
 Let $D$ be a ZX-diagram and let $E$ be its set of edges. 
 Then we refer to $E$ as the \emph{possible fault locations} of $D$. 
 A \emph{fault} $F$ on $D$ is defined as a Pauli in $\mathcal{P}^{|E|}$, indicating the action of $F$ on each edge of $D$.
 We denote $D^F$ for the ZX-diagram where we place the corresponding Pauli faults on each edge of $D$. The \emph{weight} $\operatorname{wt}(F)$ is the number of non-identity Pauli operations in $F$.
\end{definition}

\autoref{fig:faultexamplesZX} gives an example of a diagram $D$ with a total number of $11$ fault locations (i.e.\ $11$ edges) and a weight-$5$ fault on $D$.
\begin{figure}[!tbh]
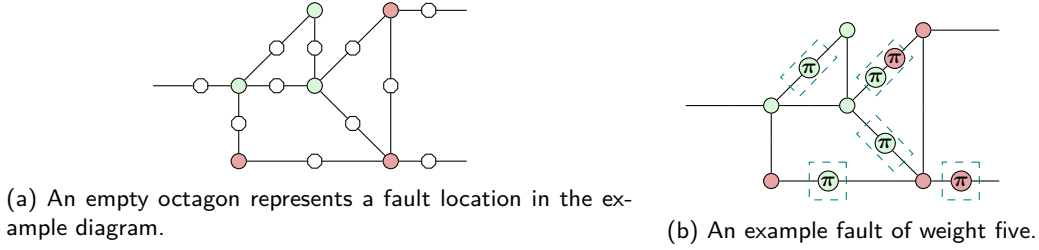

    \centering
    \begin{subfigure}{0.55\textwidth}
        \centering
        \input{figs/prelim/example-diagram-fault-locations.tikz}
        \caption{An empty octagon represents a fault location in the example diagram.}
        \label{fig:example-diagram-fault-locations}
    \end{subfigure}
    \begin{subfigure}{0.1\textwidth}
    \end{subfigure}
    \begin{subfigure}{0.35\textwidth}
        \centering
        \input{figs/prelim/example-diagram-example-fault.tikz}
        \caption{An example fault of weight five.\\}
        \label{fig:example-diagram-example-fault}
    \end{subfigure}
    \caption{An example of a diagram with a possible fault.}
    \label{fig:faultexamplesZX}
\end{figure}

One important property of a fault on a diagram is whether it is \emph{detectable}.
As noted in \cref{rem:post-selection}, we treat diagrams as post-selected on expected measurement outcomes. 
We say a fault is detectable if it makes this expected measurement outcome impossible, i.e.\@ when it makes the diagram go to zero. 
\begin{proposition}[Detectability of Faults,~\cite{rodatz2025fault}]
    \label{prop:detectability}
    Let $F \in \overline{\mathcal{P}^{|E|}}$ be a fault on a non-zero diagram $D$.
    Then $F$ is \emph{detectable} if $\intf{D^F} = 0$.
\end{proposition}

See~\cite{rodatz2025fault} for more details on detectability. Note that we can characterise undetectable faults in an alternative way that will be useful to us:
\begin{proposition}
    \label{prop:undetectable-faults-boundary}
    A fault $F$ on a non-zero diagram $D$ is undetectable if and only if there exists some fault $F'$ that only acts on the free edges of $D$ such that $D^F = D^{F'}$.
\end{proposition}
The proof of this proposition, as well as most other proofs in the main text are postponed to \cref{app:proofs}.

We can now define what it means for two diagrams or circuits to be fault-equivalent.
\begin{definition}[Fault equivalence, $w$-fault equivalence]
 Let $D_1, D_2$ be two diagrams with possible fault locations $\mathcal{F}_1$ and $\mathcal{F}_2$. 
 We say $D_1$ is \emph{$w$-fault-equivalent} to $D_2$ if and only if for all undetectable faults $F_1 \in \langle \mathcal{F}_1 \rangle$ with weight $\operatorname{wt}(F_1) < w$ there exists a fault $F_2 \in \langle \mathcal{F}_2 \rangle$ on $D_2$ with $\operatorname{wt}({F}_2) \leq \operatorname{wt}({F}_1)$ and $D_1^{F_1} = D_2^{F_2}$
  and vice versa with $(D_1,\mathcal{F}_1)$ and $(D_2,\mathcal{F}_2)$ interchanged.
 We write $D_1 \underset{w}{\FaultEq} D_2$ if $D_1$ is $w$-fault-equivalent to $D_2$.
 We write $D_1 \FaultEq D_2$ and say $D_1$ and $D_2$ are \emph{fault-equivalent} if they are $w$-fault-equivalent for all $w \in \mathbb{N}$.
\end{definition}

Intuitively, $w$-fault equivalence guarantees a correspondence between the undetectable faults of weight less than $w$ in the two diagrams. 
It ensures that for any undetectable fault on the one side, there exists an equivalent fault on the other side with at most the same weight. 
Therefore, neither diagram can achieve a specific undetectable fault with fewer fault events than the other. 
We can then define what it means when a diagrammatic construction \emph{$t$-fault-tolerantly implements} a \CAT state.

\begin{definition}[Fault-tolerant \CAT-state preparation]
    \label{def:ft-cat-state-t}
    A ZX-diagram $D$ $t$-fault-tolerantly ($t$-FT) prepares the \CAT state if it satisfies:
    \[
        \scalebox{0.8}{\input{figs/prelim/def-ft-cat-state.tikz}}
    \]
\end{definition}

Note that it is not immediately obvious that this definition recovers the usual notion of fault-tolerant \CAT-state implementation (see, e.g.,~\cite{forlivesi2025flag,peham2026optimizing}). Our model places Pauli faults on diagram edges rather than Pauli errors after circuit gates. In \cref{app:edge-flip-noise}, we discuss this in detail and prove the two definitions are equivalent for \CAT states.


Viewing fault-tolerant construction through the lens of fault equivalence leads to a new framework for fault-tolerant circuit synthesis: if a circuit is fault tolerant if and only if it is fault-equivalent to some idealised (i.e.\@ fault-free) specification, then one can start with the specification and rewrite it into a circuit. 
As long as fault equivalence is preserved at each step, we are guaranteed that the final circuit is fault-tolerant by construction. 
We present a collection of fault-equivalent rewrite rules from~\cite{rodatz2024floquetifying,rodatz2025fault} in \cref{fig:fe-rewrites} which will be useful for fault-tolerant \CAT-state preparation.

\begin{figure}[!tbh]
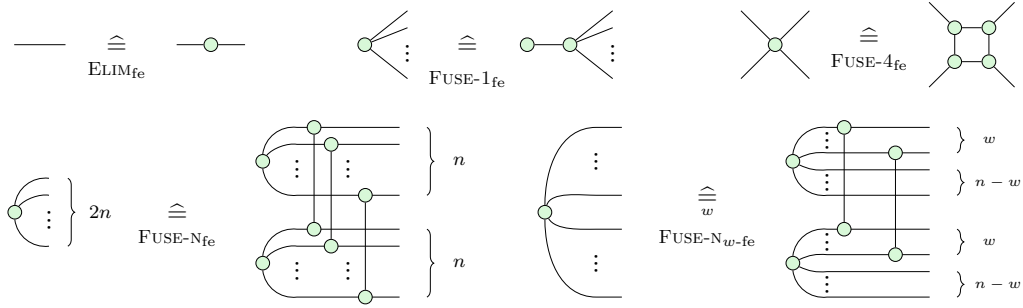

  \begin{equation*}
    \scalebox{.9}{\input{figs/prelim/fe-rewrites1.tikz}}
  \end{equation*}
  \begin{equation*}
    \scalebox{.9}{\input{figs/prelim/fe-rewrites2.tikz}}
  \end{equation*}
  \caption{A collection of fault-equivalent rewrite rules that are useful for fault-tolerant \CAT-state preparation. Here $w,n \in \N$ and $w \leq n$.}
  \label{fig:fe-rewrites}
  \vspace{-0.3cm}
\end{figure}

\section{Constructing low-depth fault-tolerant \CAT states}\label{sec:recursive-cat}

To demonstrate the utility of fault-equivalent rewrites, we will provide a simple scalable construction of $t$-FT \CAT states on an arbitrary number of qubits $n$ and fault-weight $t$ such that the gate count is $O(n)$, qubit count is $O(n)$ and circuit depth is $O(\log t)$.

First, note that the rewrite \FEwcat in \cref{fig:fe-rewrites} gives us a way to decompose a $2n$-qubit \CAT state into two $n$-qubit \CAT states connected together by $w$ transversal $ZZ$-measurements. So suppose we wish to construct $t$-FT \CAT states and that we are given some FT construction of an $n$-qubit $t$-FT \CAT state using $C$ CNOT gates. Then from \FEwcat we get a $2n$-qubit \CAT state that costs $2C+2(t+1)$ CNOT gates (each $ZZ$-measurement requires 2 CNOT gates). Doubling again gives us a $4n$-qubit \CAT state costing $2(2C+2(t+1))+2(t+1) = 4C + 6(t+1)$ CNOTs. In general we see that for a $2^kn$-qubit \CAT state we use $2^kC +2(2^k-1)(t+1)$ CNOT gates, which is linear in the number of qubits. Similarly, if the base $n$-qubit \CAT state requires $a$ ancillae, then the $2^kn$-qubit \CAT state requires $2^ka + (2^k-1)(t+1)$ ancillae. The CNOT depth of a $ZZ$-measurement is two, so after $k$ layers, the depth has increased by $2k$, which is logarithmic in the total number of qubits $2^kn$, which matches the asymptotic depth of~\cite{peham2026optimizing}. However, we can do better.

Observe that in \FEwcat only $2w$ qubits of the $2n$ total qubits are involved in a $ZZ$-measurement, leaving $2(n-w)$ qubits `free'. Then, if we use this circuit in the next doubling, we could choose to place the next layer of $ZZ$-measurements on these free qubits, since due to the symmetry between the qubits in \FEwcat it doesn't matter on which qubits we place the transversal $ZZ$-measurements. Hence, if the number of free qubits is large enough to fit this next layer of $w$ $ZZ$-measurements, we don't have to increase the depth of the circuit, but can fit it into the same layer of $ZZ$-measurements. In particular, suppose that $n\geq 2w$. Then there are $2n-2w\geq 2w$ free qubits after the first doubling. Then for the second doubling we double the number of qubits, and hence the number of free qubits, giving at least $4w$ free qubits, of which we have to use $2w$ in transversal $ZZ$-measurements, so that we are left with at least $2w$ free qubits. By induction we see that we will always have enough free qubits ($2w$) to never have to increase the depth. Hence, the CNOT depth of this construction for $t=w+1$ is a constant above the depth of the base-case $n$-qubit $t$-FT \CAT state of size $n=2(t+1)$, giving a depth of $O(\log t)$.

\begin{theorem}\label{thm:recursive-cat}
		We can efficiently construct $n$-qubit $t$-FT \CAT states using $n(1+\log (t+1)) - 2(t+1)$ CNOT gates, $n/2$ ancillae and in CNOT depth $2\log t + 2$.
\end{theorem}

\cref{fig:16q-ft} gives the smallest interesting examples of $t$-fault-tolerant \CAT-state preparation using the recursive approach. When $t$ increases, the measurement depth is increased accordingly.

\begin{figure}[!tbh]
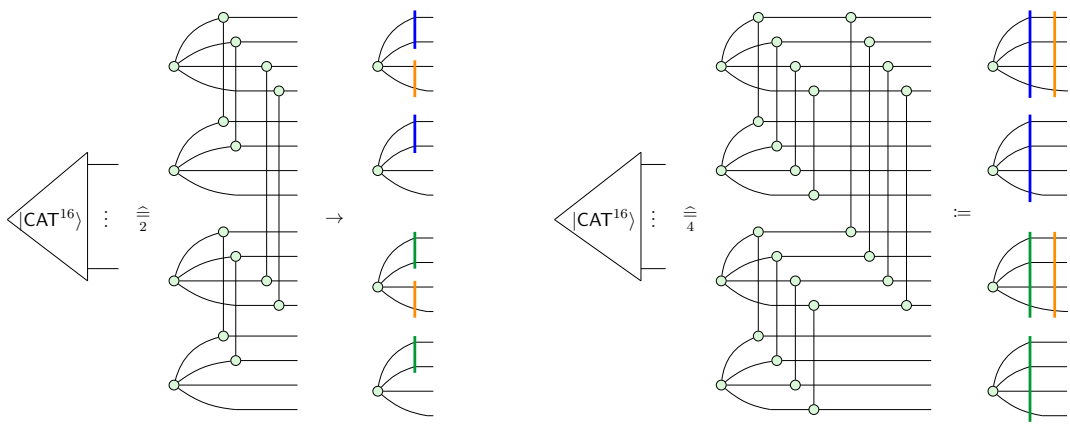

    \centering
    \begin{equation*}
         \scalebox{.65}{\input{figs/bendi/16q-ft2.tikz}} \qquad \quad
          \scalebox{.65}{\input{figs/bendi/16q-ft4.tikz}}
    \end{equation*}
    \vspace{-.4cm}
\caption{Example of a $t$-FT preparation of a $16$-qubit \CAT state for $t\in\{1,3\}$.
On the left there are enough free qubits to parallelise the $ZZ$-measurements, while on the right multiple layers are needed. Coloured lines represent the individual layers of $ZZ$-measurements coming from \FEwcat.
}
\label{fig:16q-ft}
\vspace{-.3cm}
\end{figure}

\section{Constructing \CAT states from 3-ary graphs}
\label{sec:circ-extract}

In the previous section we found a construction of fault-tolerant \CAT states from ZX-diagrams that only contain $Z$-spiders of arity 3, meaning the spider has three legs. It turns out that we can more generally consider diagrams of this type to find other construction for \CAT-state preparations.

\begin{definition}
	A \emph{$Z$-graph} is a phase-free ZX-diagram consisting solely of 3-ary $Z$-spiders with no inputs (so all boundary spiders are connected to outputs).
	The \emph{underlying graph} of the $Z$-graph is the graph consisting of a vertex for each spider, where we don't include the output wires in this graph.
\end{definition}

\begin{figure}
  \begin{equation*}
    \scalebox{0.8}{\input{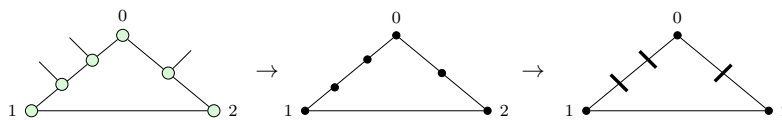}}
  \end{equation*}
  \vspace{-.4cm}
  \caption{An example of a Z-graph, its underlying graph, and its underlying marked graph.}
  \label{fig:graph-and-marked-graph}
  \vspace{-.3cm}
\end{figure}

It turns out we can analyse which $Z$-graphs can implement fault-tolerant \CAT states by considering properties of the underlying graph. As a very simple example, we can see that the underlying graph must be connected.
If the graph were disconnected, then there would be two groups of outputs across which the state would be a product state.
Any connected component which does not contain any outputs is equivalent to a scalar and can be removed.

If the $Z$-graph $t$-fault-tolerantly prepares a \CAT state for $t\geq 1$, then it turns out that each boundary vertex can only be connected to a single boundary.

\begin{lemma}\label{lem:unique-boundary}
	Let $D$ be a $Z$-graph $t$-FT preparing a \CAT state with at least 4 output wires for $t\geq 1$. Then each vertex in $D$ is connected to at most one boundary.
\end{lemma} 

The underlying graph of a $Z$-graph contains both the internal and boundary vertices, where because of \cref{lem:unique-boundary} each boundary vertex is reduced to a 2-ary vertex, since its output wire is not included. An example of the relationship between $Z$-graph and underlying graph is presented in \cref{fig:graph-and-marked-graph}, where we also show the corresponding marked graph that will be introduced in \cref{sec:ft-cat-graph}.

To construct a quantum circuit implementing a \CAT state from a $Z$-graph or its underlying graph, we need some more information.
In particular, we need to know which spiders should end up on the same qubit wire and in which order they should appear on that wire.
For this it turns out to be sufficient to consider a \emph{rooted spanning tree} of the underlying graph, which is a tree containing all the vertices in the graph, where the root is a designated `starting' vertex. The direction from the root to the output wires gives the order in which spiders should appear on the same qubit.

\begin{definition}
  Let $G$ be an underlying graph of a $Z$-graph. A \emph{spider-ordering tree} for $G$ is a rooted spanning tree of $G$ where the root is a 3-ary vertex, and all the leaves are 2-ary vertices.
\end{definition}



We give an algorithm for going from a $Z$-graph with spider-ordering tree to a circuit implementing a \CAT state in \cref{proc:circuit-extraction}.

\begin{procedure}\label{proc:circuit-extraction}
To extract a CNOT circuit from a graph $G$ with spider-ordering tree $T$ we proceed as follows.
\begin{enumerate}
\item Convert the underlying graph $G$ back into a $Z$-graph $D$ by writing each 3-ary vertex as a $Z$-spider, and each 2-ary vertex as a boundary $Z$-spider, and connect the spiders in the natural way.
\item Starting from the root vertex of $T$, do the following. If there is a child which is a leaf and is a boundary spider, put this spider on the same qubit line as $T$ (i.e.~same vertical position) to the right of the parent. If there is no such a child, pick an arbitrary child and do the same. Put all the other children on different qubit lines and to the right of the parent. Label all children as `placed'. Now, for each child of the root, consider that child's unplaced children, and follow the same rule of picking one child to be put on the same qubit line, with other children being placed on different qubit lines. Since all leafs are boundary spiders, each qubit line ends in a free wire.
\item For the root spider of $T$, use $\text{Elim}_{\text{fe}}$ and $\text{Fuse-1}_{\text{fe}}$, so that now each spider has exactly one connection out of its own qubit line.
\item For each edge going out of a qubit line, introduce $X$-spiders using $\text{Elim}_{\text{fe}}$ and $\text{Fuse-1}_{\text{fe}}$ to write this connection as a CNOT either followed by a post-selection or preceded by an ancilla preparation.
\item Now interpret all the ZX-components as components of a quantum circuit following \cref{fig:mappings}.
\end{enumerate}
Each of these steps is demonstrated in \cref{fig:circ-extract-steps}.
\end{procedure}

\begin{figure}[!t]
		\begin{equation*}
			\scalebox{0.8}{\input{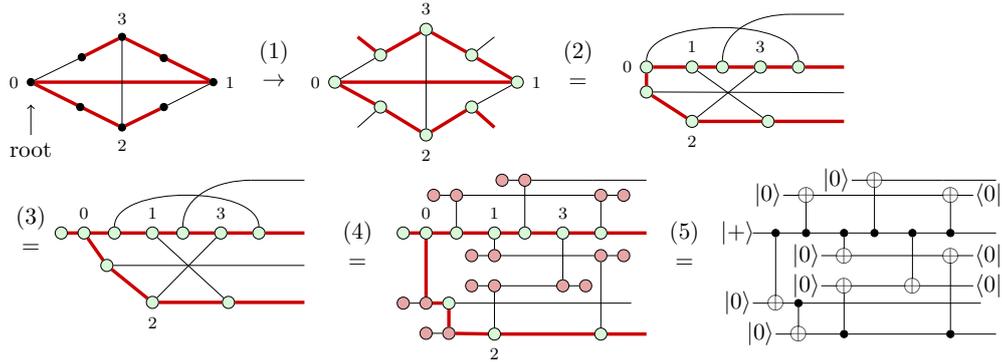}}
		\end{equation*}	
		\caption{A demonstration of all the rewrite steps in \cref{proc:circuit-extraction}. The red line denotes the tree.}
		\label{fig:circ-extract-steps}
    \vspace{-.4cm}
\end{figure}

\begin{proposition}\label{prop:circuit-extract-cnot-count}
		Let $G$ be a graph with $k$ vertices and a spider-ordering tree. Then \cref{proc:circuit-extraction} produces a CNOT circuit from it using $k+1$ CNOT gates.
\end{proposition}
\begin{proof}
	The ZX-diagram of $G$ has $k$ 3-ary $Z$-spiders. Each $Z$-spider introduces one CNOT gate, except for the spider corresponding to the root of the tree, which gets unfused into a pair of 3-ary spiders in step 2 of \cref{proc:circuit-extraction}. Each 3-ary spider is then involved in precisely one CNOT gate.
\end{proof}
Note that the associated $Z$-graph of $G$ has $k$ spiders, and hence that the CNOT count is exactly 1 more than the spider count. This turns out to be optimal, as we will see in the next section. Note that a connected graph always has a spanning tree. In practice we find that we can always find a spanning tree with the appropriate additional constraints.

It is further worth emphasising that any circuit or gadget used to optimally implement  $t$-FT $n$-qubit \CAT state can also optimally fault-tolerantly measure a weight $n$ stabiliser in a distance $t$ code using the noise model with faults on edges.
This is done by changing step 4 to connect the qubit line to a corresponding data qubit instead of an ancilla output.
The main line can also now start or end with such a CNOT (or CZ, for $Z$ stabilisers) connecting it to the data qubits.
This will also be fault-tolerant for CSS codes, but under more realistic circuit-level noise, more care is needed to make such a measurement fault-tolerant.

\section{Lower bounds for circuits of \CAT states}
\label{sec:lower-bounds}
We saw in the previous section that we can construct \CAT state circuits from diagrams consisting solely of $Z$-spiders. This raises the question of whether this is the best way to start building these circuits. We will show in this section that if our goal is to minimise the CNOT count that we indeed cannot do better. We do this by showing in \cref{sec:reduction-to-z-spiders} that \emph{any} CNOT circuit implementing a \CAT state gives rise to a `skeleton' of 3-ary $Z$-spiders which would allow us to construct a \CAT state using at least as many CNOT gates as using \cref{proc:circuit-extraction}. In \cref{subsubsec:opt-cnot-h}, we derive the same lower bound on CNOT count of any CNOT+Hadamard circuit.

A next question to ask is then what the minimal size is of a 3-regular graph that implements an $n$-qubit \CAT state that is fault-tolerant up to some fault weight $t$. This we will do by showing that suitable graphs must not be easily `cuttable' by removing edges in \cref{sec:ft-cat-graph}.

\subsection{Optimality of diagrams of Z-spiders}\label{sec:reduction-to-z-spiders}
Here, we show that any CNOT circuit implementing a \CAT state has gate cost at least that of a circuit produced from \cref{proc:circuit-extraction}. To do this, let's first define a general type of circuit suitable for implementing \CAT states. 

\begin{definition}
	A \emph{generalised CNOT circuit} consists of CNOT gates, ancillae preparations in the $\ket{0}$ and $\ket{+}$ states, measurements in the $Z$- or $X$-basis and Pauli $Z$ or $X$ corrections based on these measurement outcomes satisfying the following constraints:
	\begin{itemize}
		\item It may specify post-selection criteria, a list of `allowed' measurement outcomes, where we throw away the outcome of the circuit and try again if the measurement outcomes are not in the allowed set. The all 0 measurement outcome passes the post-selection criteria.
		\item It must be \emph{deterministic}, meaning that if the parity corrections are executed according to the observed measurement outcomes, that the overall linear map implemented will always be the same. The all 0 measurement outcome requires no Pauli correction.
	\end{itemize}
    \label{def:gen-cnot}
\end{definition}


\begin{procedure}
    A generalised CNOT circuit $C$ can be represented by a phase-free ZX-diagram $D$ by translating the circuit post-selected onto the all 0 measurement outcome into a diagram following the table of \cref{fig:mappings}. We refer to this diagram as the \emph{associated diagram} of the circuit. 
    \label{def:associated-diagram}
\end{procedure}

\begin{remark}
    The associated diagram $D$ contains 3-ary $Z$- and $X$-spiders (corresponding to CNOT gates) and 1-ary spiders (corresponding to ancilla preparations and postselected measurement outcomes). It is equal to the linear map the circuit deterministically implements.
\end{remark}    

\begin{procedure}
    Given a generalised CNOT circuit $C$, the \emph{spider-skeleton} $D'$ of $C$ is a phase-free ZX-diagram obtained from the associated ZX-diagram by repeatedly fusing and copying all 1-ary spiders and removing the resulting identity 2-ary spiders.
    \label{def:spider-skeleton}
\end{procedure}

By \cref{def:spider-skeleton}, the spider-skeleton only contains 3-ary spiders and hence can be seen as a 3-ary open graph with vertices labelled by either $Z$ or $X$. \cref{fig:spider-skeleton} provides examples where rewrite rules are applied locally to get rid of 1- and 2-ary spiders without reducing the fault distance. All proofs of the statements below can be found in \cref{subsubsec:opt-cnot}.

\begin{lemma}\label{lem:CNOTs-to-spider-count}
	Let $C$ be a generalised CNOT circuit that $t$-fault-tolerantly prepares an $n$-qubit \CAT state $\ket{\CAT^n}$ with $n>1$. Let $D'$ be its spider-skeleton. Then $D' \wFaultEq{t+1} \ket{\CAT^n}$. If $D'$ contains $k$ $Z$-spiders, then $C$ contains at least $k+1$ CNOT gates. 
\end{lemma}

Note that the spider-skeleton of the circuit produced by \cref{proc:circuit-extraction} is exactly the original $Z$-graph. By \cref{prop:circuit-extract-cnot-count}, the CNOT count of the circuit constructed from such a $Z$-graph with a suitable spanning tree is precisely 1 more than the spider count of the circuit. As a consequence of \cref{lem:CNOTs-to-spider-count}, this is optimal.

The $Z$-graphs from the previous section only contain $Z$-spiders, while spider-skeletons can also contain $X$-spiders. It turns out that we can construct a `reduced skeleton' from a spider-skeleton only consisting of $Z$-spiders, which will generate the same \CAT state to the same level of fault-tolerance. Following \cref{def:reduced-spider-skeleton}, \cref{fig:reduced-spider-skeleton} illustrates the process of obtaining a reduced skeleton $S$ from a spider-skeleton $D'$.

\begin{procedure}\label{def:reduced-spider-skeleton}
	Given the spider-skeleton $D'$ of a \CAT state circuit $C$, we produce the \emph{reduced skeleton} $S$ in the following way:
	First find the smallest stabilising Pauli web $W$ in $D'$ of the all $X$ Pauli string on the boundary (which must exist since this Pauli string is a stabiliser of the \CAT state~\cite{rodatz2024floquetifying}). 
	Then, construct $S$ by starting with $D'$ and dropping all edges that are not highlighted by $W$ and then removing all free-floating spiders (those not connected to an output via some connected path) and removing all identity spiders of degree two using $\text{Elim}_{\text{fe}}$.
\end{procedure}

\begin{proposition}\label{prop:Zonly}
Let $S$ be the reduced skeleton of a generalised CNOT circuit $C$ implementing an $n$-qubit \CAT state with fault-tolerance up to weight $t$.
Then $S$ is a $Z$-graph which has at most as many $Z$-spiders as the spider-skeleton of $C$, and $S$ is $t$-FE to the $n$-qubit \CAT state.
\end{proposition}

\begin{theorem}\label{thm:CNOT-optimality}
	Let $C$ be generalised CNOT circuit containing $k$ CNOT gates, which $t$-fault-tolerantly prepares the $n$-qubit \CAT state.
	Then we can efficiently find a $Z$-graph containing at most $k-1$ spiders, which also $t$-fault-tolerantly prepares the $n$-qubit \CAT state.
\end{theorem}

By this theorem, if we have a lower bound on the size of a $Z$-graph that $t$-fault tolerantly prepares an $n$-qubit \CAT state, this will also give a lower bound on the size of any generalised CNOT circuit that $t$-fault-tolerantly prepares the $n$-qubit \CAT state. In the next section we will see how we can derive such lower bounds.

\subsection{Graphs of Z-spiders as fault-tolerant CAT states}
\label{sec:ft-cat-graph}

To find an implementation of a \CAT state that uses as few CNOT gates as possible, we should try to find a $Z$-graph with as few spiders as possible, while still retaining the desired level of fault-tolerance. Since $Z$-graphs correspond to graphs we can give a different perspective on this problem: to prepare a $t$-FT $n$-qubit \CAT state, we wish to minimise the total number of vertices. In order to see how to do this, let's consider how we can view undetectable faults on a $Z$-graph in graph-theoretic terms. Proofs of the statements below can be found in \cref{subsec:ft-construction}.

\begin{proposition}\label{prop:undetectable-fault-is-cut}
Let $D$ be a $Z$-graph with underlying graph $G = (V,E)$.
A fault $F$ on $D$ is undetectable if and only if it corresponds to a \emph{cut} of $G$: a set of edges that when removed from $G$ disconnects it into two components.
\end{proposition}

Hence, to show a $Z$-graph $t$-FT prepares a \CAT state, it suffices to look at faults that correspond to cuts of the graph cutting along at most $t$ edges.
As it can never hurt the fault tolerance of our circuit if a fault occurs on an output edge, we can just consider cuts on the underyling graph.


It turns out to be nice to consider a type of graph where boundary spiders are treated as a modification on an edge instead of a separate vertex.

\begin{definition}\label{def:marked-graph-diagram}
    A \emph{marked graph} $(G,M)$ is a graph $G$ with a list of marks $M$, where each edge is allowed to be \emph{marked} one or more times. A mark is represented by an orthogonal dash through that edge. Let $D$ be a $Z$-graph implementing a fault-tolerant \CAT state on at least 4 qubits. We define its \emph{underlying marked graph} as the underlying graph where we replace each boundary vertex and its incident edges by a marked edge between the two neighbours of the boundary vertex.
\end{definition}

See \cref{fig:graph-and-marked-graph} for an example of getting a marked graph from a Z-graph.
Note that by \cref{lem:unique-boundary}, each boundary vertex is connected to a unique boundary, and hence in the underlying graph it has arity 2, so that \cref{def:marked-graph-diagram} is well-defined. If boundary vertices are connected to each other in the $Z$-graph, that will result in multiple marks on the same edge regardless of the order in which we replace the boundary vertices. Since the underlying marked graph removes all vertices coming from the boundaries, it removes all the 2-ary vertices from the graph, so that the underlying marked graph is completely 3-regular.

We can easily get back a $Z$-graph from a marked 3-regular graph: replace each vertex by a 3-ary $Z$-spider, keeping the edges as connections between spiders. Second, replace each mark by a boundary $Z$-spider and replace the edge it was on with two edges coming from the new spider. Hence, instead of working with $Z$-graphs, we can just work with marked 3-regular graphs instead.

In order to $t$-FT prepare the \CAT state, after cutting the graph along $f\leq t$ edges, at least one of the two subgraphs must contain at most $f$ output wires. Otherwise, if both contain $>f$ outputs, then after pushing the faults out to the boundary, we would have a higher-weight fault than we started with. This then gives the definition for what types of graphs we are looking for.

\begin{definition}\label{def:good-cut}
	Let $D$ be a $Z$-graph and $(G,M)$ its underlying marked graph. Let $A,B\subseteq G$ be the components of $G$ induced by cutting along the edges going between $A$ and $B$ (where $A$ and $B$ can contain open edges or fragments of edges if an edge in $G$ was cut multiple times). We say this is a \emph{good cut} if at least one of $A$ and $B$ contains at most $t$ marks (as we never allow a cut at a mark or a vertex, the marks $M$ are distributed fully among $A$ and $B$). We say $(G,M)$ is \emph{$t$-robust} if every cut on $t$ or fewer edges is a good cut.
\end{definition}

This definition allows us to sum up the fault-tolerance condition as follows:

\begin{proposition}
	Let $D$ be a $Z$-graph and $(G,M)$ its underlying marked graph. Then $D$ is $t$-FE to a \CAT state if and only if $(G,M)$ is $t$-robust.
\end{proposition}

We are hence interested in finding graphs for which each small cut (in at most $t$ locations) is a good cut, while also trying to minimise the number of vertices. If we group marks too closely, they will be too easy to cut off, so this forces us to include a minimum number of vertices in our $Z$-graph. We will consider the ratio between the number of qubits in the \CAT state and the number of vertices in the graph. 

\begin{definition}\label{def:vertex-ratio}
		Let $D$ be a $Z$-graph and $(G,M)$ its underlying marked graph. If $G$ has $v$ vertices and $M$ contains $n$ marks, we say the \emph{vertex ratio} of $G$ is $r(G) = v/n$. Let $r_t^n$ be the minimum vertex ratio over all $t$-robust graphs with $n$ marks, and let $r_t = \liminf\limits_{n\to\infty} r_t^n$ be the \emph{optimal vertex ratio at fault-weight $t$}.
\end{definition}

\begin{proposition}
    Any generalised CNOT circuit that $t$-fault-tolerantly prepares the $n$-qubit \CAT state contains at least $n(r_t^n+1)+1$ CNOT gates. 
    \label{prop:cnot-optimal}
\end{proposition}

As we will see, for many values of $n$ and $t$, the vertex ratio $r_t^n$ is often \emph{equal} to $r_t$, showing that this definition gives a useful estimate for computing resource requirements for \CAT states for a given fault-distance.

We can also use these bounds to say something about the number of measurements needed in the circuit.
In a connected 3-regular graph with $V$ vertices and $E$ edges, we have $2E=3V$.
We can then also compute the \emph{cyclomatic number} of the graph, the number of linearly independent loops, as $E-V+1=\frac{V}2+1$.
This number precisely measures the number of independent detecting regions we have and thus is a lower bound on the number of flags or measurements that we need to perform in our extracted circuit.
Thus, if we have a graph with $n$ marks and a vertex ratio of $r_t=V/n$, we will require $\frac{r_t}{2}n+1$ flags at a minimum. It is however not yet clear whether this is also a lower bound across all circuits, as \cref{thm:CNOT-optimality} only shows that Z-graphs are just as good as generalised CNOT circuits when it comes to low CNOT count, but doesn't say anything about the number of measurements in the original CNOT circuit.

We can calculate $r_t$ exactly for several small values of $t$.
For $t=1$ the calculation of the optimal vertex ratio $r_1$ is simple. We just consider the graph with no vertices and a single edge that makes a loop.
Note that this loop is not a real edge, since it is not between two internal vertices and would violate $2E=3V$, but we can treat it as part of edges between boundary vertices and add marks to it anyway.
This loop cannot be cut into two components by a single fault, meaning that we can put any number of marks on this edge and it is still $1$-robust. On the level of the $Z$-graph this corresponds to:
\[
\scalebox{.8}{\input{figs/ft-rewrite/cycle.tikz}}
\]
We conclude that $r_1 = 0$, since the underlying graph contains no vertices.

For higher $t$ the discussion becomes more complicated.
It is clear that certain properties are desirable for optimizing the vertex ratio. If the graph has a small `nonlocal' cut, which cuts the graph into two roughly equally-sized parts, then each can only contain a limited number of marks, limiting the vertex ratio. Hence, the graph should be well-connected, preventing such cuts, only allowing local parts of the graph to be cut off by a small cut.

\begin{definition}\label{def:small-big-cut}
A \emph{local cut} cuts a diagram into connected components, with at least one of them being a tree. A \emph{nonlocal cut} cuts a diagram into connected components, with none of them being a tree.
\end{definition}  

The best graphs for constructing fault-tolerant \CAT states are those that do not contain small nonlocal cuts.

\begin{remark}\label{rem:no-cycles}
If a graph admits no nonlocal cuts of size up to $t$, it cannot have any cycles of length $t$ or fewer, as otherwise the at most $t$ edges leaving the cycle could be cut, which would cut off the vertices in the cycle, making a big cut. Hence, the \emph{girth} of the graph, the length of the shortest cycle, must be at least $t+1$.
\end{remark}

\begin{lemma}
For $t=2$, the optimal vertex ratio is $r_2 = \frac13$.
\label{lem:vertex-ratio-r2}
\end{lemma}   
\begin{proof}
First we show that $r_2 \geq \frac13$. Suppose towards contradiction that $r_2 < \frac13$ and consider a graph containing at least 6 marks. There must then on average be more than three marks per vertex, and since in a 3-regular graph, there are three edges for every two vertices, there must on average be more than two marks per edge. By the pigeonhole principle, there exists an edge that has three marks.
We then cut this edge `twice', one on the left of the marks, and once on the right on the marks. We are then left with a subgraph consisting of a single open edge and 3 marks, which does not satisfy the good cut property, i.e.~there is a weight 2 fault that propagates to a weight 3 fault.


Next, we show that $r_2 \leq \frac13$ by giving an explicit construction of a family of graphs. 
\begin{equation}
    \scalebox{0.8}{\input{figs/flag-by-construction/c2-construction.tikz}} \qquad \quad 
    \scalebox{0.7}{\input{figs/flag-by-construction/c2-construction-graph-only.tikz}}
\label{eq:c2-construction}    
\end{equation}
We see that there are no 2 different edges that can be chosen so that the graph is cut. 
As there are two marks on each edge, we prevent the problem of the bad cut above.
There is one vertex for every three marks, so $r_2 \leq \frac13$.
Together with the above, we conclude $r_2=\frac13$ and that this bound can be achieved exactly for infinitely many numbers of marks.
\end{proof}  

We can prove similar results on $r_t$ for larger $t$, though the proofs become increasingly complicated as larger cuts have to be considered.

\begin{theorem}\label{thm:vertex-ratio}
	The following are the optimal vertex ratios for various $t$: $r_1 = 0$, $r_2 = \frac13$, $r_3 = \frac23$, $r_4 = \frac56$ and $r_5 = 1$. Furthermore we have $r_\infty\coloneqq\liminf\limits_{n\to\infty}\sup_tr_t^n \leq 12$.
	Moreover, for each of these cases there exists an infinite family of graphs that achieves these optimal ratios.
\end{theorem}

A bound on $r_\infty$ is achieved by constructing a Ramanujan graph of degree 10, adding two output edges to each vertex, and then fault-equivalently rewriting the degree 12 vertices to \CAT states with 24 degree-3 vertices. From this, $r_\infty\leq\frac{24}2=12$. Since there is a family of such Ramanujan graphs of increasing size $V_k$ which have at most a constant size difference $V_{k}-V_{k-1} = O(1)$ between them, this also gives us a way to construct \CAT states up to arbitrary levels of fault-tolerance by `rounding up' to the nearest graph in this family.

\begin{corollary}\label{cor:constant-CNOT-overhead}
    There exists a constant $C$ such that any $n$-qubit \CAT state can be constructed fault-equivalently using at most $C\cdot n$ three-legged spiders. In particular, we can construct $t$-FT $n$-qubit \CAT states with $t$ arbitrarily high using $O(n)$ CNOT gates.
\end{corollary}

The way we prove the optimal vertex ratios for $t=3,4,5$ is to view the 3-regular graphs as being built out of trees that are glued together where marks can only occur at the glueing points. As $t$ increases, a larger number of trees has to be considered. We conjecture that the optimal graphs continue to be built in this `glued trees' manner, which gives us a conjecture on the value of $r_t$ as $t\rightarrow \infty$. For the details see \cref{app:vertex-ratio-conjecture}.

\begin{conjecture}
    The optimal vertex ratio $r_t$ as $t$ increases tends to $r_t = 2$.
\end{conjecture}

In \cref{sec:circ-extract} we saw that we can turn a 3-regular graph into a circuit by identifying a Hamiltonian path.
This puts all the vertices on a single line, creating a deep circuit.
We can create a much shallower circuit by parallelising the vertices onto distinct qubit lines.
This increases CNOT count, as each start of a path results in 1 additional CNOT.
In the extreme case we can put all the vertices and marks on their own triplets of qubits, creating many 3-qubit \CAT states in depth 2.
Each connection between vertices then corresponds to a Bell basis measurement, raising our depth to 3.
We can make this more efficient by performing a CNOT between two qubits directly, as long as it is the first operation applied to either of them, which it can be for at least $\frac{(3r_t+2)n}{5}$ of the vertices.

\begin{theorem}\label{thm:depth-4}
  Let $(G,M)$ be a $t$-robust marked graph with $n$ marks achieving the optimal vertex ratio $r_t$.
  Then we can construct a circuit $t$-FT preparing the $n$-qubit \CAT state using $\left(\frac{29r_t+26}{10}\right)n$ CNOT gates with CNOT depth 3 and using at most $\left(\frac{12r_t+8}5\right)n$ ancillae.
\end{theorem}

\section{Constructing optimal \CAT states}\label{sec:cat-state-generation}

Having established the minimal size of the graphs that are $t$-robust, we now address the construction of graphs that saturate these bounds.
Constructions for values of $n$ that are not multiples of the vertex ratio can be made by increasing $n$, using the optimal graph, then removing some marks, but this might not be optimal, so we can use numerical calculations for larger $t$ and for `in-between' values of $n$.

Directly solving the $t$-robustness condition of \cref{def:good-cut} is generally hard, as it couples the problem of finding suitable graphs with finding good markings.
In practice, we want to decouple these tasks into sufficient distinct problems that can be solved separately. 

\begin{proposition}
    Let \(G\) be a 3-regular graph, $M$ a set of markings, and $T$ a rooted spanning tree on $G$.
    Then, this yields a CNOT-optimal circuit implementing a $t$-fault-tolerant \CAT state if the following conditions hold:
    \begin{itemize}
        \item \(G\) admits no nonlocal cut of size at most $t$.
        \item \(T\) is a spider-ordering tree of \(G\).
        \item $M$ satisfies the necessary local conditions so that each cut of size at most $t$ is a good cut.
        \item If \(t > 2\), then each edge is marked at most once.
        \item $(G,M)$ has the optimal vertex ratio implied by \cref{thm:vertex-ratio}.
    \end{itemize}
    \label{prop:sufficient-construction}
\end{proposition}

With these sufficient conditions, we can logically separate the search for good graphs, markings, and spanning trees.
In general, these problems appears to be computationally hard in the worst case: finding markings for $t=7$ is equivalent to finding a total dominating set on the line graph of $G$, which is NP-complete, and finding graphs without nonlocal cuts is likely equally hard.
In practice, however, we can solve each problem for any practically interesting size and beyond.

To construct 3-regular graphs without nonlocal $t$-cuts, we employ a randomized hill-climbing algorithm.
Since verifying the absence of this property is computationally expensive, we first optimise two efficient proxy metrics, girth and spectral gap, to generate high-quality candidates.
The algorithm proceeds in three phases: \begin{enumerate}
    \item \textbf{Girth Optimisation:} We iteratively apply local perturbations (double edge swaps) to prioritize increasing the girth.
    As noted in \cref{rem:no-cycles}, a girth of at least $t+1$ is a necessary condition to prevent nonlocal $t$-cuts induced by small cycles.
    \item \textbf{Spectral Optimisation:} Once the girth is large enough, we continue applying perturbations to maximise the spectral gap, the difference between the first and second eigenvalues of the graph's  Laplacian.
    A large spectral gap implies a large Cheeger constant, which is a good proxy for high connectivity.
    \item \textbf{Verification:} Once the specified threshold for the spectral gap is reached, we do the expensive verification for nonlocal $t$-cuts.
    If none exist, the graph is accepted; otherwise, we resume local perturbations while maintaining the girth and spectral properties.
\end{enumerate}

\noindent This procedure is detailed in \cref{alg:construction} in \cref{app:cat-state-gen}.

To verify the absence of nonlocal $t$-cuts, we encode the search for a counter-example as a SAT instance.
Recall that a cut is nonlocal if neither component is a tree, meaning both must contain at least one cycle.
To detect this, we rely on the property that a graph contains a cycle if and only if it contains a non-empty subset of vertices inducing a subgraph with minimum degree at least 2.
We introduce boolean variables representing the partition assignment, alongside auxiliary variables identifying such a \emph{witness subset} within each partition.
The constraints are as follows:
(1) the cut size is at most $t$;
(2) every vertex in a witness subset has at least two neighbors within that same subset;
and (3) both partitions contain a non-empty witness subset.
If this instance is satisfiable, the graph contains a nonlocal cut and is rejected.
The exact formulas for the SAT instance are detailed in \cref{subsec:nonlocal-cut-SAT}.

Once a valid graph is found, we find a valid marking set \(M : E \to \{0,1\}\).
We aim to maximise the number of marked edges while respecting the local conditions for $t$-fault-tolerance derived in \cref{subsec:ft-construction}.
Since this problem involves maximising an objective function subject to complex logical dependencies, we formulate it as a Maximum Satisfiability (MaxSAT) problem.
We provide the specific encoding details in \cref{subsec:markings-SAT}. The codes for constructing \CAT-state preparation circuits can be found in the project’s GitHub repository~\cite{cat2026}.

Finally, we look for a spanning tree.
Our algorithm constructs this in two distinct phases.
First, we establish a basic spanning forest by using exclusively unmarked edges.
Then, we iteratively merge these disjoint components until only a single spanning tree remains.
To ensure the final tree is well-balanced, we employ a greedy heuristic that selects merging edges to minimise the overall diameter of the newly merged components.

At each step, the algorithm evaluates all valid edges bridging two distinct trees and calculates the potential new diameter based on the eccentricities of the connecting nodes.
Specifically, when merging a tree $A$ and a tree $B$ via an edge connecting node $u$ (in $A$) to node $v$ (in $B$), the resulting diameter is exactly:
\[
  \operatorname{diam}(A \cup B) = \max(\operatorname{diam}(A),\ \operatorname{diam}(B),\ \operatorname{ecc}_A(u) + 1 + \operatorname{ecc}_B(v))
\]
The candidate edge that results in the smallest overall diameter is chosen.
The procedure returns a well-balanced tree with a minimal number of marked edges, both of which promote better error rates in the extracted circuits.

\section{Simulation results}
\label{sec:simulation-results} 
We construct marked graphs that match our theoretical lower bounds for every target distance $t \in \{2, \dots, 7\}$ and for \CAT-state sizes up to $n=50$. For larger \CAT-states (i.e.\ $50<n\leq100$ when $t\leq 5$), we additionally provide all constructed circuits in our GitHub repository~\cite{cat2026}. In every instance found, the resulting circuits achieve the theoretical minimum CNOT count, providing strong evidence that SpiderCat scales effectively while saturating the theoretical bounds. To our knowledge, SpiderCat enables the largest systematic generation of CNOT-optimal fault-tolerant \CAT-state preparation circuits to date.

We then benchmark SpiderCat against two existing fault-tolerant \CAT-state preparation strategies: Flag-at-Origin (FAO)~\cite{forlivesi2025flag} and MQT (the Munich Quantum Toolkit)~\cite{peham2026optimizing}. In \cref{subsec:resource-overhead}, we compare the resource overheads of the circuits produced by each method (flag and CNOT counts), and in \cref{subsec:performance}, we compare their simulated performance.

\subsection{Comparing resource overheads}
\label{subsec:resource-overhead}

\cref{fig:benchmark2,fig:benchmark3} summarise the flag and CNOT overheads of the \CAT-state preparation circuits produced by SpiderCat, MQT, and FAO. Since the CNOT count scales linearly with the flag count, we focus our discussion on the flag-count heatmap in \cref{fig:benchmark2}.

\cref{fig:benchmark2} presents the flag overhead across a range of $(n,t)$ values, with the \CAT-state size $n$ shown on the horizontal axis and the fault distance $t$ on the vertical axis. Each cell is annotated and color-coded by the number of flags required, with a few special cases. White cells indicate parameter pairs $(n,t)$ for which the given method provides no construction. Gray cells mark entries that we do not need to consider, since a construction for $(n,t)$ is implied by an existing construction for $(n,t-1)$. A blue cell labeled * indicates that, although the formula predicts the optimal flag count, we proved that this value is unattainable via exhaustive search.
Brown cells labeled $\dagger$ denote cases where a solution based on \cref{prop:sufficient-construction} cannot exist.
This is because no 3-regular graph exists with the required girth and number vertices, as the $(t+1)$-cage~\cite{Tutte_1947} has more vertices than demanded by the method.
However, an optimal solution might still exist using a graph of lower girth and an alternative construction.
Finally, cyan cells labeled \# indicate that the search timed out, leaving the existence of a solution undetermined.

For SpiderCat, the map is densely populated over the explored range (roughly $8\leq n \leq 50$ and $3\leq t \leq 8$), with flag counts increasing smoothly as either $n$ or $t$ grows. This indicates that the overhead grows in a controlled and scalable way even for larger \CAT states and higher target distances. MQT, in contrast, shows noticeably larger flag counts in the high-$n$ and high-$t$ regime (see the darkest colors and largest annotations near the upper-right of its heatmap), showing a steeper growth in resource overhead as the target distance increases. Moreover, MQT does not provide constructions that cover the full explored $(n,t)$ range, leaving gaps in the heatmap at larger $n$ and/or $t$. Finally, FAO covers a much smaller feasible region (most entries in its heatmap are absent for larger $t$ and $n$), indicating the limited scalability of this approach compared to SpiderCat and MQT.

\begin{sidewaysfigure}
    \includegraphics[width=\textwidth]{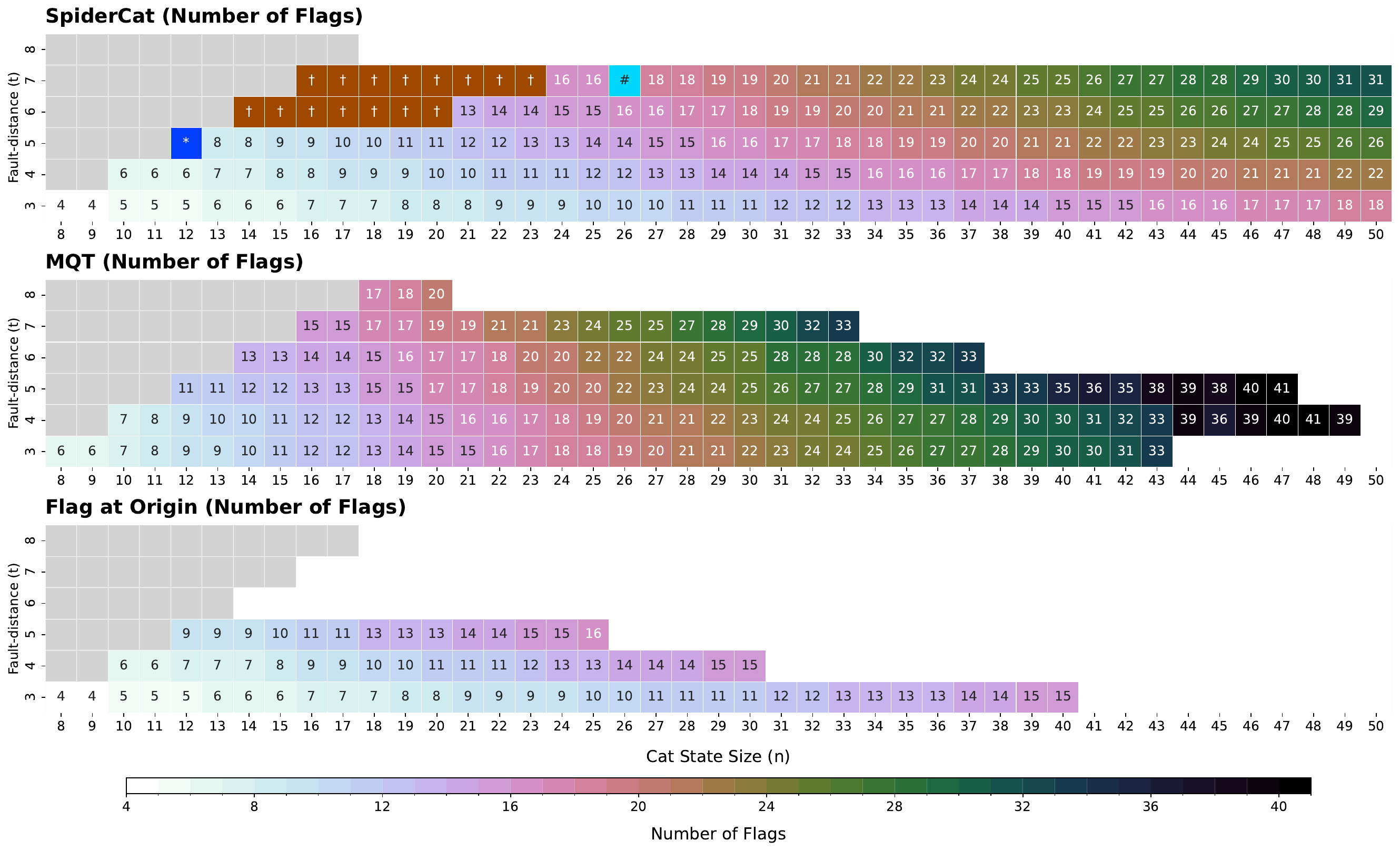}
    \caption{Comparing the number of flags for different fault-tolerant \CAT-state preparation methods. Missing entries are places where a given method has no solution. Note that SpiderCat can produce \CAT states for almost all $n \leq 100$ when $t \leq 5$.
    For $(n, t)=(12, 5)$, no solution exists with the predicted optimal value of 7.
    For $(n, t)=(26, 7)$, our computation timed out.
    For values marked $\dagger$, no graph with girth $>t$ exists that matches the predictied optimal vertex ratio, but other graphs might still yield optimal solutions.}
    \label{fig:benchmark2}
\end{sidewaysfigure}

\begin{sidewaysfigure}
    \includegraphics[width=\textwidth]{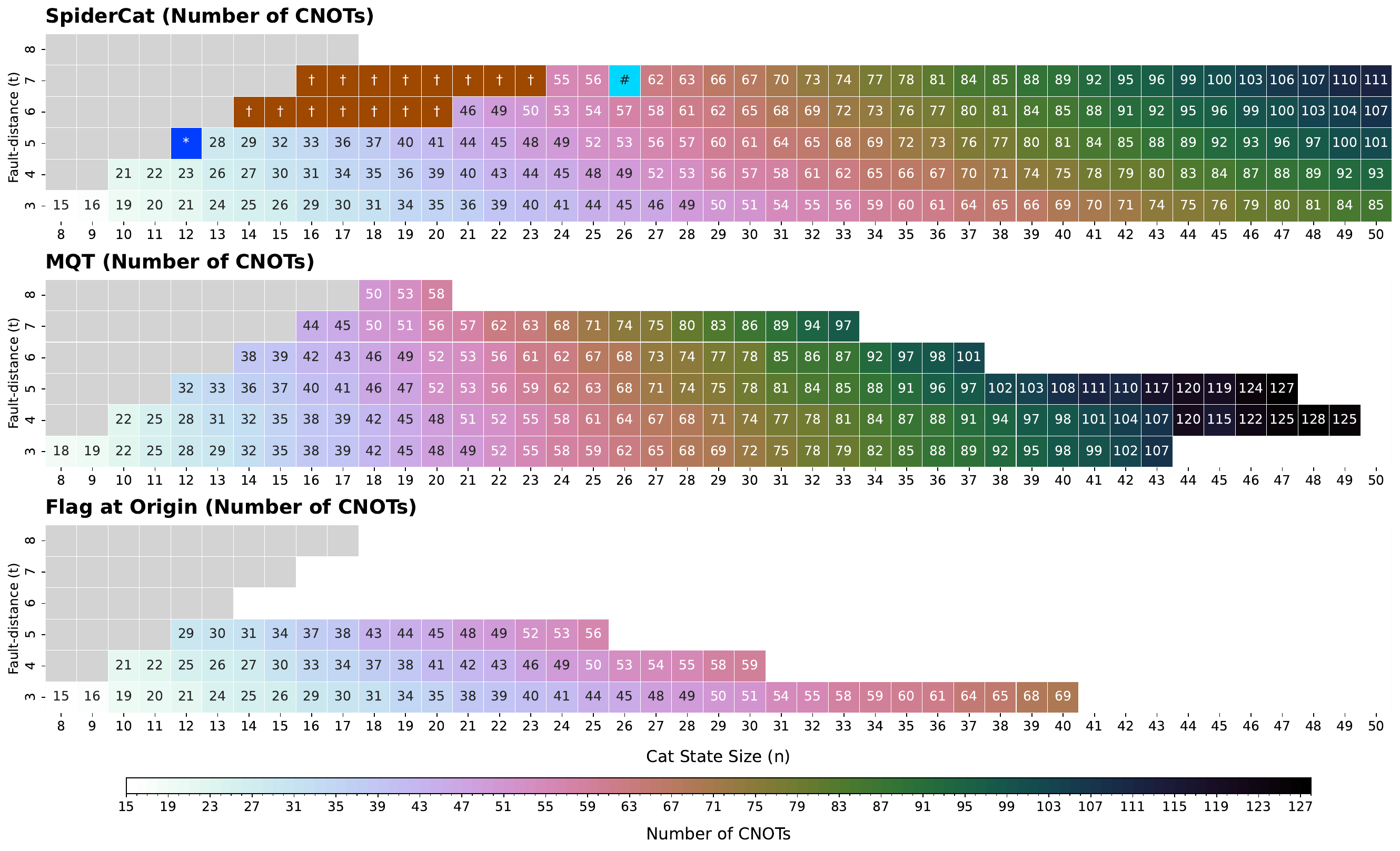}
    \caption{Comparing the number of CNOT gates for different fault-tolerant \CAT-state preparation methods. Missing entries are places where a given method has no solution. Note that SpiderCat can produce \CAT states for almost all $n \leq 100$ when $t \leq 5$. For $(n, t)=(12, 5)$, no solution exists with the predicted optimal value of 7.
    For $(n, t)=(26, 7)$, our computation timed out.
    For values marked $\dagger$, no graph with girth $>t$ exists that matches the predictied optimal vertex ratio, but other graphs might still yield optimal solutions.}
    \label{fig:benchmark3}
\end{sidewaysfigure}

\begin{figure}
    \makebox[\textwidth][c]{
  \begin{subfigure}{0.59\textwidth}
    \caption{Methods Comparison for $t=3$}
    \includegraphics[width=\textwidth]{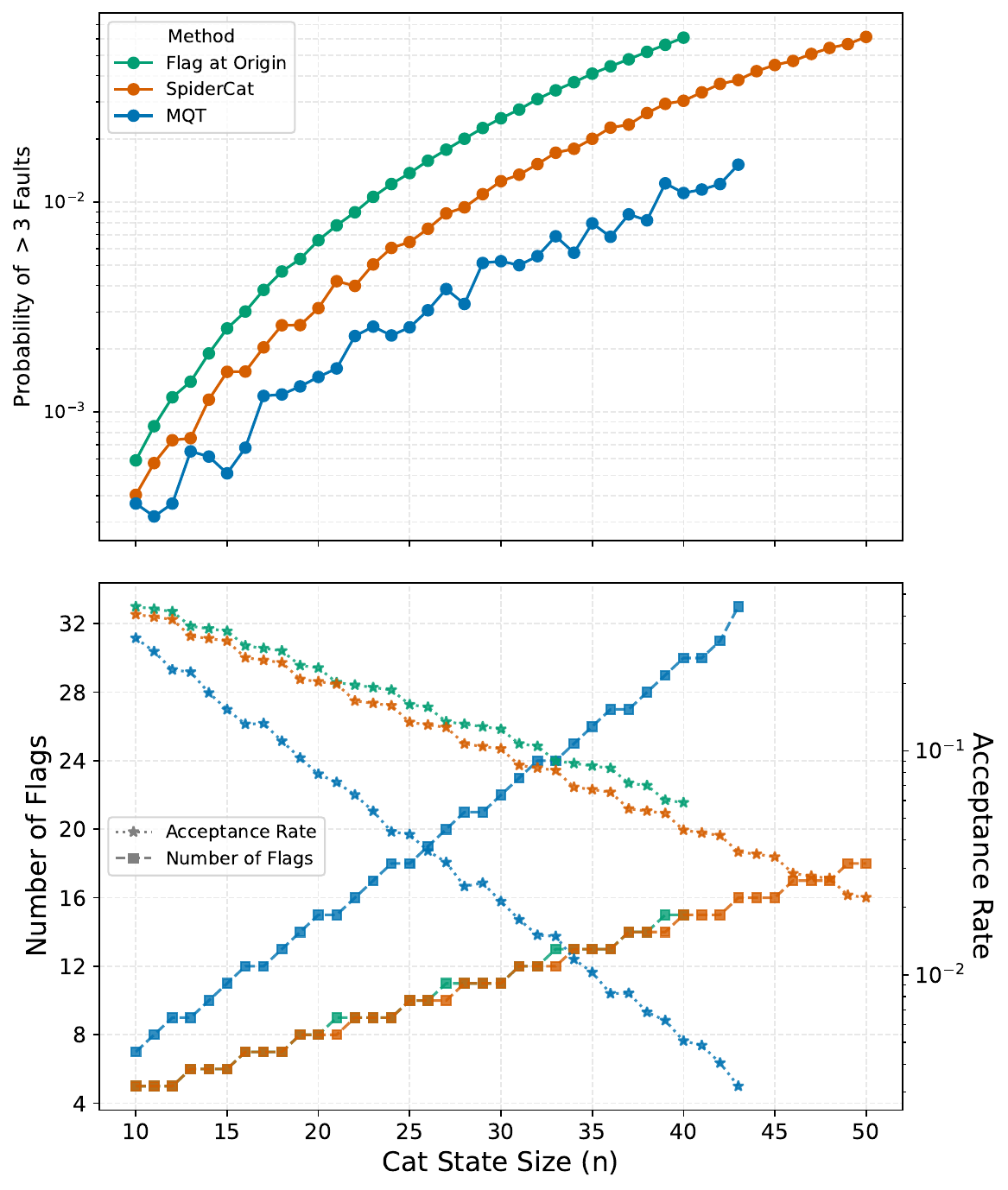}
  \end{subfigure}
  \begin{subfigure}{0.59\textwidth}
    \caption{Methods Comparison for $t=4$}
    \includegraphics[width=\textwidth]{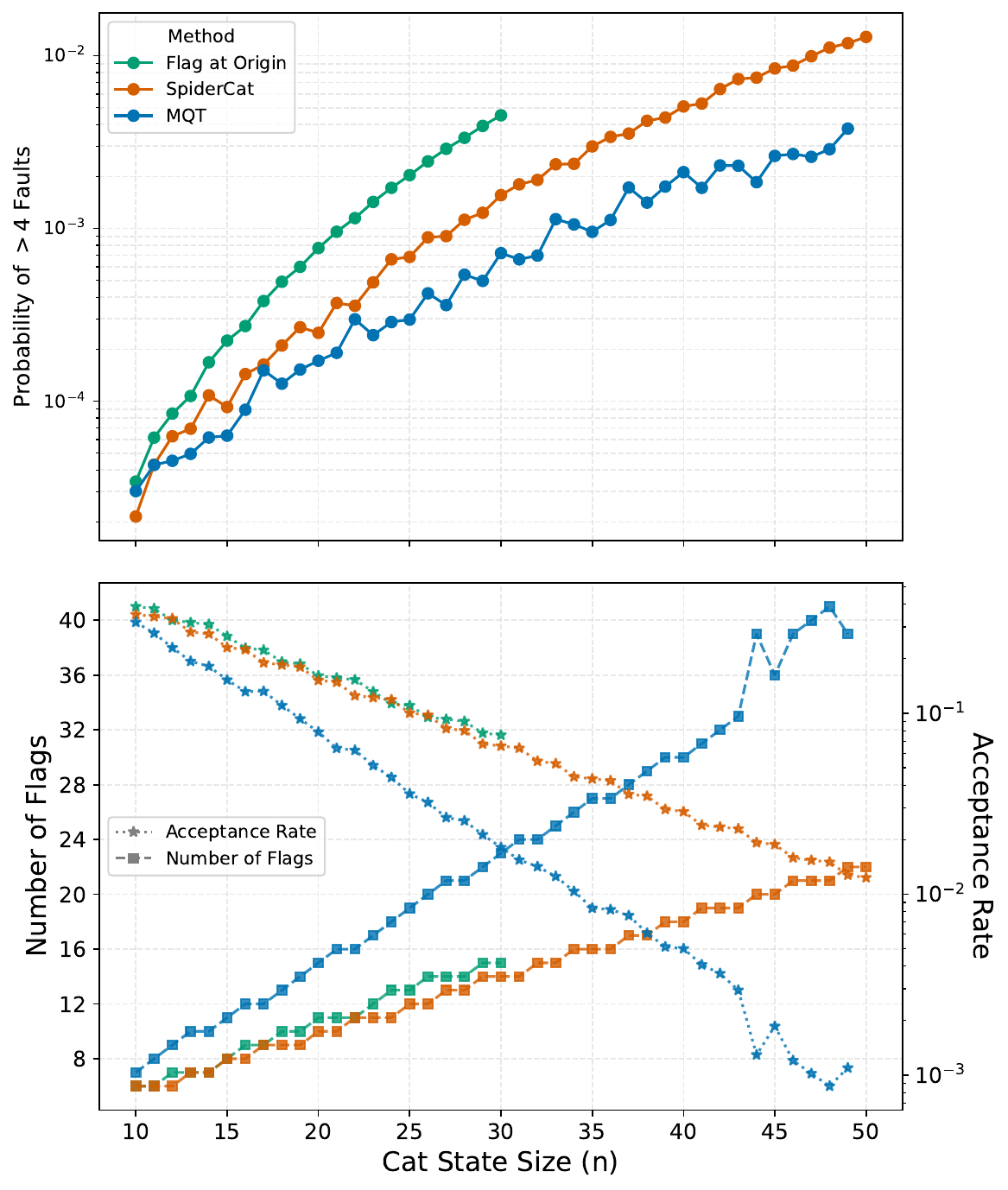}
  \end{subfigure}}
    \makebox[\textwidth][c]{
  \begin{subfigure}{0.59\textwidth}
    \caption{Methods Comparison for $t=5$}
    \includegraphics[width=\textwidth]{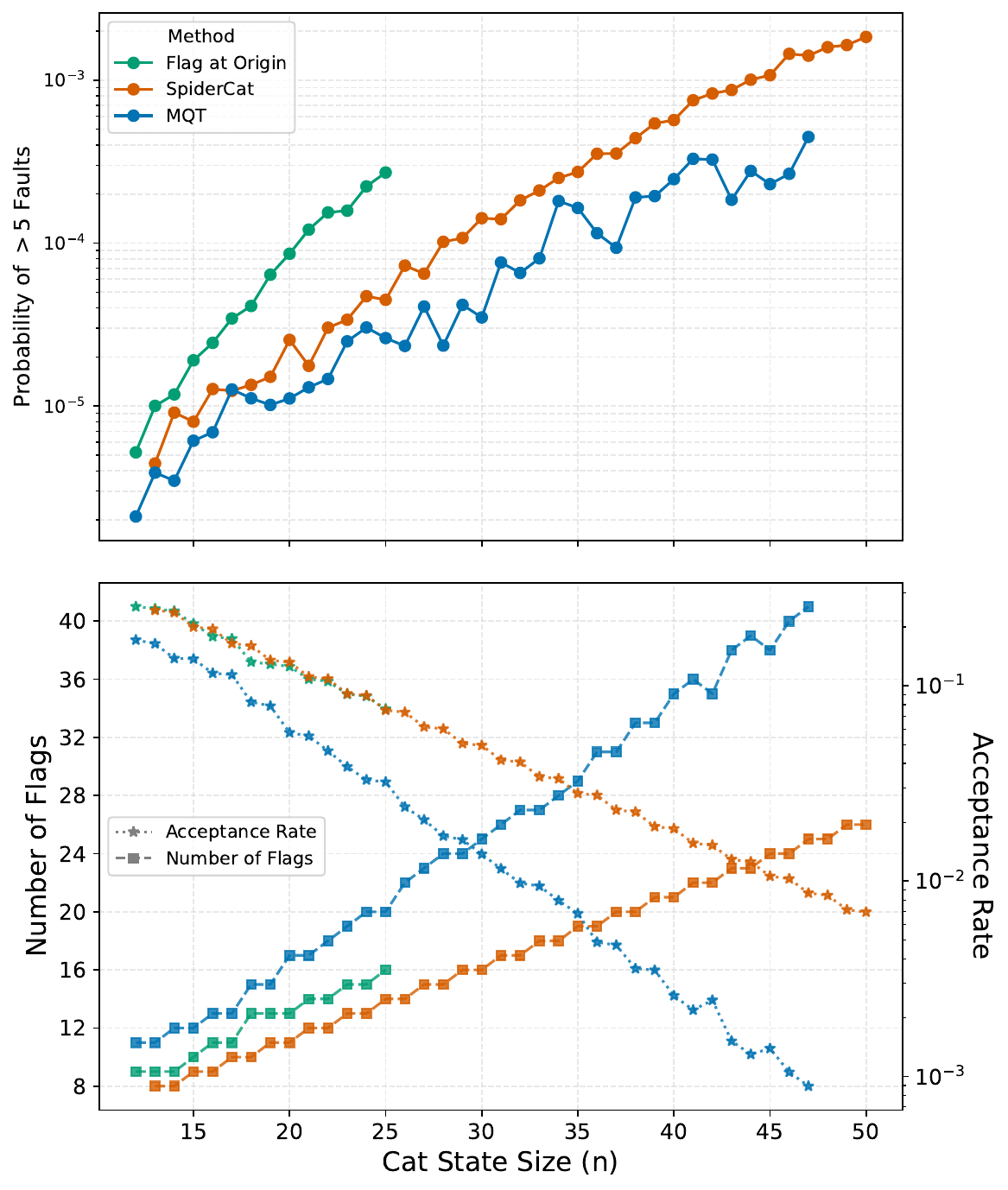}
  \end{subfigure}
  \begin{subfigure}{0.59\textwidth}
    \caption{Methods Comparison for $t=6$}
    \includegraphics[width=\textwidth]{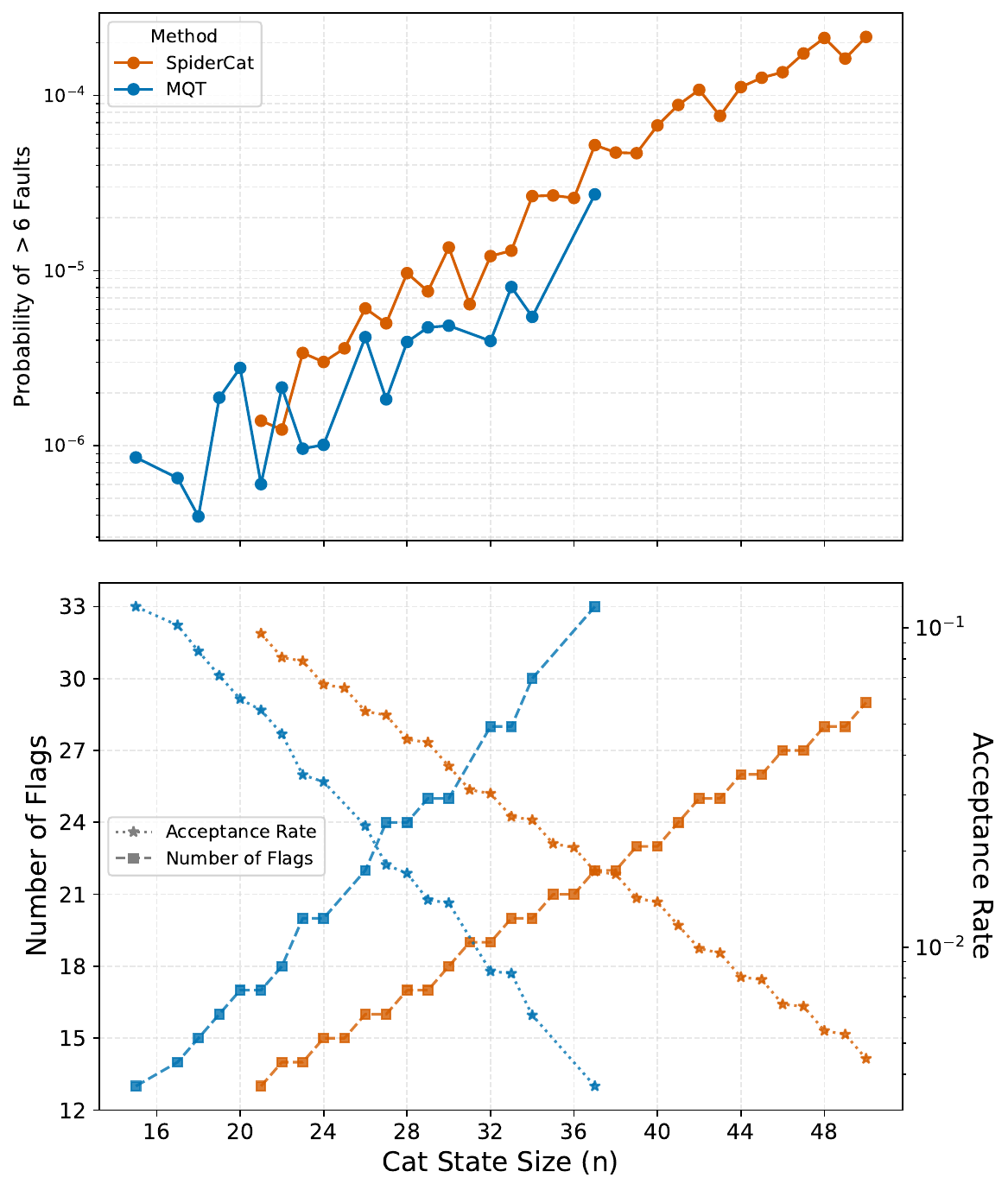}
  \end{subfigure}}
  \caption{Comparing different fault-tolerant \CAT-state preparation methods.
  For each value of $t$, the top plot shows the probability of $>t$ faults on a \CAT state prepared with SpiderCat, MQT, and Flag at Origin.
  The lower plots show the number of flags needed to prepare the state and the acceptance rate under postselection.
  Note that we observe a power law between these two variables.}
  \label{fig:benchmark-curves}
\end{figure}

\subsection{Comparing performance in simulation}
\label{subsec:performance}

To compare the performance of the examined methods, we conducted Monte Carlo simulations using Stim~\cite{Gidney2021stimfaststabilizer}.

\paragraph{Noise model}
We assume circuit-level depolarizing noise with error probabilities $\pTwo$ for two-qubit gates and $\pSPAM$ for state preparations and measurements.
For a given physical error probability $\pPhis$, we set $\pTwo \coloneqq \pPhis$ and $\pSPAM \coloneqq \frac{2}{3} \pPhis$ in accordance with the error model used in~\cite{peham2026optimizing}.

\paragraph{Results}
We set $\pPhis \coloneqq 0.05$ and take $5{\small,}000{\small,}000 * t$ samples for each circuit.
\cref{fig:benchmark-curves} shows the performance of SpiderCat, MQT, and FAO as the \CAT-state size $n$ increases for target fault-tolerance levels $t=3,4,5,6$. For each $t$, the top plot in the corresponding panel reports the simulated logical error probability (i.e.\@ the probability of $>t$ faults) as a function of $n$, while the bottom plot tracks the number of flags together with the acceptance rate (i.e.\@ the probability that verification passes and the state is kept).

Across all $t \in \{3,4,5,6\}$, \cref{fig:benchmark-curves} suggests a strong correlation between the number of flag qubits and the observed fault-tolerant performance. Increasing the flag count implements stronger verification by adding more checks that can detect high-weight faults, but it simultaneously increases the circuit size and measurement overhead, creating more fault locations and more opportunities to trigger rejection. This trade-off is reflected in the bottom plots of each panel: as the number of flags grows with $n$, the acceptance rate decreases, often by orders of magnitude. The logical error probability then reflects the net effect of these tradeoffs: adding more flags can catch more bad error patterns, but after a while it mostly just causes more runs to be rejected (lower acceptance rate) without a matching drop in the failure probability of the top plot. Among all three methods, MQT require rapidly increasing numbers of flag qubits as $n$ grows, leading to a steep drop in acceptance rate even though their $>t$-fault probabilities continue to increase with $n$. SpiderCat, by contrast, strikes a more favorable balance: it maintains lower $>t$-fault probabilities while using fewer flags and keeping the acceptance rate much higher.

Compared to SpiderCat, both MQT and FAO demonstrate weaker scalability as the \CAT-state size $n$ and the target distance $t$ increase, but for different reasons. For MQT, the bottom plots show that the number of flags grows steeply with $n$ and the acceptance rate falls by orders of magnitude, reaching very low values at large $n$ (especially for $t=4$ and $t=5$). This indicates rapidly worsening throughput, as improved verification comes primarily through heavier postselection. FAO becomes impractical even sooner: the top panels show consistently higher $>t$-fault probabilities than the other methods, and its brute-force construction does not scale to larger $t$, which is why it is not shown for $t=6$.

\section{Conclusion}
\label{sec:conclusion}
We have presented a provably optimal strategy for fault-tolerant \CAT-state preparation, with formal lower bounds reached by explicit circuit constructions.
We did this by expressing circuits as graphs and therefore shifting the problem from studying undetectable fault to cuts of graphs. 
The resulting fault-tolerant circuits are both asymptotically and practically efficient, making them suitable for near-term experimental implementation. 
Generalising this perspective, we aim to study other fault-tolerant circuit compilation tasks, such as state preparation and implementing logical Clifford operators for a given error-correcting code.




In practice, \CAT-state preparation involves considerations beyond minimising circuit depth, CNOT count, or ancilla qubit count. Depending on the available architecture, full fault tolerance may be too costly, and across different parameter regimes and physical error rates, different constructions (i.e.\@ those presented here and in prior work) may achieve the best logical error rates. In simulations, we have found that even when fixing the underlying graph and the CNOT count, the logical error rate was still quite sensitive to the exact circuit extraction algorithm used.
In particular, the constructions in this paper do not obtain the lowest possible logical error rates (the MQT method outperforms SpiderCat, though at a significantly increased CNOT and postselection cost) due to these exact details and sensitivities.
We have further discovered a family of additional constructions that we believe 13-fault-tolerantly prepares \CAT states of any size. In simulations, such constructions achieve a significantly lower logical error rate than other state-of-the-art methods. The detailed analysis and proofs, as well as finding improved methods for circuit extraction, are left for future work.

\paragraph*{Acknowledgements}
We thank Mackenzie Shaw, Aleks Kissinger, Dmitry Paramonov, and Michael Vasmer for many helpful ideas, suggestions, and discussions. ABK, SML, BR, and JvdW gratefully acknowledge the Dagstuhl Seminar 25382, Quantum Error Correction Meets ZX-Calculus, where many of the ideas presented in this work were first discussed and developed. We also acknowledge the use of automated language tools for editorial assistance. The circuit diagrams in this paper were typeset using TikZiT~\cite{tikz}. JvdW wishes to acknowledge support from an NWO Veni grant.
ABK and BP are supported by the Engineering and Physical Sciences Research Council grant number EP/Z002230/1, `(De)constructing quantum software (DeQS)''.
BR and RY acknowledge support from Simon Harrison via the Wolfson Harrison UK Research Council Quantum Foundation Scholarship.

\bibliography{cnotsyn}

\appendix

\section{More details on the ZX calculus}
\label{appendix:zx-intro}
ZX-diagrams define linear maps.
To get the linear map defined by a ZX-diagram, we apply the \emph{interpretation function}.
The interpretation function is defined compositionally: first on the basic generators (the $Z$ and $X$ spiders), and then extended to arbitrary diagrams via composition.

\begin{definition}[Z and X spiders]
\begin{align*}
\textit{$Z$-spider:} \qquad
\input{figs/prelim/z-spider.tikz}
\quad &\coloneqq \quad
\ket{0}^{\otimes n}\!\bra{0}^{\otimes m}
+ e^{k\pi i}\ket{1}^{\otimes n}\!\bra{1}^{\otimes m}, \\[8pt]
\textit{$X$-spider:} \qquad
\input{figs/prelim/x-spider.tikz}
\quad &\coloneqq \quad
\ket{+}^{\otimes n}\!\bra{+}^{\otimes m}
+ e^{k\pi i}\ket{-}^{\otimes n}\!\bra{-}^{\otimes m}.
\end{align*}
\end{definition}

In addition to spiders, ZX-diagrams may contain \emph{identities}, \emph{swaps}, \emph{cups}, and \emph{caps}, which are interpreted as follows:
\[
\tikz[tikzfig]{ \draw (0,0) -- (2,0); }
\ :=\ 
\sum_i \ket{i}\!\bra{i}
\qquad
\input{figs/prelim/swap.tikz}
\ :=\ 
\sum_{ij} \ket{ij}\!\bra{ji}
\qquad
\input{figs/prelim/cup.tikz}
\ :=\ 
\sum_i \ket{ii}
\qquad
\input{figs/prelim/cap.tikz}
\ :=\ 
\sum_i \bra{ii}.
\]

More complex diagrams are formed by composing simpler ones, either sequentially or in parallel.

\begin{definition}[Sequential and parallel composition]
Let $n,m,k,l \in \mathbb{N}$.
Let
\[
D_1 : (\mathbb{C}^2)^{\otimes n} \to (\mathbb{C}^2)^{\otimes m}, \quad
D_2 : (\mathbb{C}^2)^{\otimes m} \to (\mathbb{C}^2)^{\otimes k}, \quad
D_3 : (\mathbb{C}^2)^{\otimes k} \to (\mathbb{C}^2)^{\otimes l}
\]
be ZX-diagrams.

The \emph{sequential composition} $D_2 \circ D_1 : (\mathbb{C}^2)^{\otimes n} \to (\mathbb{C}^2)^{\otimes k}$ is defined by
\[
\interp{\input{figs/prelim/D2-circ-D1.tikz}}
\ \coloneqq \
\interp{\input{figs/prelim/D2.tikz}} \circ \interp{\input{figs/prelim/D1.tikz}}.
\]

The \emph{parallel composition}
$D_1 \otimes D_3 : (\mathbb{C}^2)^{\otimes (n+k)} \to (\mathbb{C}^2)^{\otimes (m+l)}$
is defined by
\[
\interp{\input{figs/prelim/D1-tensor-D3.tikz}}
\ \coloneqq \
\interp{\input{figs/prelim/D1.tikz}} \otimes \interp{\input{figs/prelim/D3.tikz}}.
\]
\end{definition}

To interpret larger diagrams, we can decompose them into basic components:
\[\scalebox{0.9}{\input{figs/prelim/composition-example.tikz}}\]
While every ZX-diagram can be interpreted as its underlying linear map this way, in practice, this is not usually necessary. We can translate directly between quantum circuits and ZX-diagrams, and then compose them according to the circuit structure. \cref{fig:mappings} illustrates how various quantum states and Pauli operators are represented as Pauli ZX-diagrams. \cref{fig:pauli-zx} presents a set of Pauli rewrite rules that will be used to transform ZX-diagrams corresponding to the same linear map. For further background on the Pauli ZX-calculus, see~\cite{kissinger2022phase,KissingerWetering2024Book}.

We encourage readers to consult~\cite{KissingerWetering2024Book,van2020zx} for a comprehensive overview of the complete calculus.

A fundamental property of the ZX calculus is that the interpretation is independent of how a diagram is decomposed into generators or in which order compositions are performed. Which gives to the property that \emph{only connectivity matters} for determining the linear map of a diagram.

\begin{figure}[!tbh]
  \centering
  \begin{equation*}
  		\scalebox{0.8}{\input{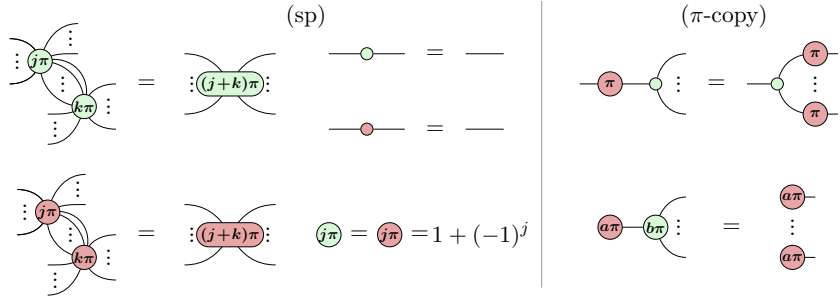}}
  \end{equation*}
  \caption{Useful rules of the Pauli ZX-calculus. For each rewrite rule, its \LHS and \RHS correspond to the same linear map up to a known scalar factor which will not be relevant in this paper. These rules remain valid up to wire bending and under interchange of $X$- and $Z$-spiders.}
  \label{fig:pauli-zx} 
\end{figure}


\section{Equivalence of edge-flip noise to circuit noise for \CAT states}\label{app:edge-flip-noise}

We associated to each Pauli fault happening on a set of edges a \emph{weight} with the assumption that higher-weight faults are less likely to occur. However, this is just an assumption of a \emph{noise model} and not one that necessarily corresponds to a realistic one.
For example, a fault may have different weight under phenomenological noise than it might have under circuit-level noise. 
As such, a noise model can be seen as a function from faults to weights. 
We can make this concrete by defining a specific noise model, which is the one we use throughout this paper: 
\begin{definition}[Edge-flip noise]
    Under edge-flip noise, the weight of a fault is the number of edges it flips, i.e.\ $\mathcal{N}(F) = |F|$.
\end{definition}

A different noise model might for instance consider a weight-2 Pauli occurring directly after a CNOT gate as a weight-1 fault, since a CNOT gate could introduce correlated errors. In general these noise models are different and would give rise to a different notion of fault-equivalence.

In particular, in~\cite{forlivesi2025flag,peham2026optimizing} they say a circuit fault-tolerantly prepares the \CAT state up to some weight $w$, if all faults of weight at most $w$, under circuit noise, are either detectable or create at most as many Pauli flips on the output qubits .

It turns out that for \CAT states in particular, the edge-flip noise model and the circuit-level noise model are equivalent.
\begin{proposition}
    \label{prop:cat-state-noise-models}
    A Pauli ZX-diagram $D$ fault-equivalently prepares the \CAT state under edge-flip noise if and only if it fault-equivalently prepares the \CAT state under circuit-level noise.
\end{proposition}
\begin{proof}
    ($\Rightarrow$) To prove this, we will first observe that for we can consider any fault to be some combination of $X$ flips and $Z$ flips by decomposing $Y$ faults into their $X$ and $Z$ components. 
    More formally, let $F$ be an undetectable fault on some diagram $D$, let $F_X, F_Z$ be faults consisting of purely $X$ flips and $Z$ flips respectively, such that $F_X F_Z = F$. 
    As $F$ is undetectable, by \cref{prop:undetectable-faults-boundary}, there exist faults $F'$ that purely act on the boundary of $D$. 
    We can once again decompose $F'$ into $F_X'$ and $F_Z'$. 
    As the \CAT state stabilises any pair of $ZZ$ flips on its boundary, we can find an equivalent fault to $F_Z'$ that acts on at most one edge. 
    Thus, the only relevant factor for whether $F$ has an equivalent fault on the specification of weight at most $|F|$ is the $X$ component of $F$. 
    If there exists an equivalent $X$ fault on the boundary of weight at most $|F|$, then, we can push the remaining $Z$ fault to one of the edges acted upon by that fault. 
    Therefore, for the \CAT state, we only have to care about the $X$ component of any fault. 

    Since Pauli ZX-diagrams correspond to quantum circuits consisting of CNOTs and Pauli rotation, we can observe that all faults of weight one under circuit-level noise correspond to single edge flips.
    The only edge flip which does not have an immediate circuit analogue is a single CNOT gate fault that can create Pauli faults on all of its outputs. 
    Thus, an $XX$ fault after a CNOT is considered a fault of weight $1$. 
    However, as an $XX$ fault after a CNOT is the same as an $X$ fault on the control of the CNOT, such a fault is equivalent to a single edge flip before the CNOT. 
    Thus, all relevant $X$-type, circuit-level noise faults are equivalent to individual CNOTs and, thus, for the $X$-type faults $\mathcal{N}_{circuit-level}(F) = \mathcal{N}_{edge-flip}(N)$ and therefore, fault equivalence under edge-flip noise implies fault equivalence under circuit-level noise. 
    In particular, any circuit-level noise fault $F$ of weight $|F|$ corresponds to a edge-flip noise fault $G$ such that $|G| = |F|$. 
    But then, if $D$ is fault-equivalent to the idealised specification, $G$ has a corresponding fault $G'$ on the specification of at most the same weight. 
    As $D^F = D^G = D^{G'}$, $G$ is also equivalent to $F$ and $\operatorname{wt}(G') \leq \operatorname{wt}(G) = \operatorname{wt}(F)$.
    Therefore, the fault equivalence is satisfied.
        
    ($\Leftarrow$) This direction is immediate as circuit-level noise contains all edge-flips as atomic faults. 
    Thus, fault equivalence under circuit-level noise trivially implies fault-equivalence under edge-flip noise. 
\end{proof}

\section{Proofs}\label{app:proofs}

In this appendix several proofs will make use of \emph{Pauli webs}, which are a type of decoration on ZX-diagrams that allow one to easily reason about detectability. 
For a good introduction to Pauli webs see~\cite[Section 2.4]{rusch2025completeness}.

\subsection{Proofs of \texorpdfstring{\cref{sec:preliminaries}}{Section 2}}

\begin{proof}[Proof of \cref{prop:undetectable-faults-boundary}]
    ($\Leftarrow$) Let $F$ be an undetectable fault. By~\cite{rusch2025completeness}, we know that its space-time equivalence class is uniquely defined by the Pauli webs it anticommutes with.
    As it is detectable, it must commute with all detecting region. 
    Therefore, $F$ only anticommutes with logical, stabilising and co-stabilising Pauli webs. 
    But then, by stabiliser theory, we know that there exist destabilisers on the free edges that exactly anticommute with the same logical, stabilising and co-stabilising Pauli webs as $F$, forming a fault $F'$. 
    But then, by~\cite{rusch2025completeness}, we must have $D^F = D^{F'}$ where $F'$ only acts on the boundary.
    
    ($\Rightarrow$) If there exists a fault $F'$ that only acts on the boundary of $D$ and $D^F = D^{F'}$, then we know that $F'$ only applies a Pauli rotation to $D$. Since, by assumption, $D \not= 0$, this means that $D^{F'} \not= 0$. 
    As $D^F = D^{F'} \not= 0$, $F$ is undetectable. 
\end{proof}

\subsection{Proofs of \texorpdfstring{\cref{sec:recursive-cat}}{Section 3}}

\begin{proof}[Proof of \cref{thm:recursive-cat}]
   Write $w=t+1$. For simplicity we will assume that $n = 2^k w$ for some $k$. If that is not the case, then these numbers would be off by some small constants. Note that in the recursive construction, the first layers where the \CAT states have fewer than $w$ qubits, we fully connect the qubits with $ZZ$-measurements. Let's first consider the stages after we have gotten our $w$-FT $w$-qubit \CAT states.

   To increase the size further we need to only use $w$ $ZZ$-measurements. Working backwords from $n$, the last layer has $2w$ CNOTs, giving 2 \CAT states with each having $n/2$ qubits. Then we require $2\cdot 2w$ CNOTs for decomposing each of those. We see that for $k$ layers we need $2\cdot (2^{k-1}-1)w$ CNOTs. We need to decompose until we get $n\cdot 2^{-k} = w$. Note that $k = \log \frac{n}{w}$. Hence, the CNOT count is $2\cdot (2^{k-1}-1)w = 2\cdot (\frac12 \frac{n}{w} -1)w = n - 2w$. At that point we need $2^k = \frac{n}{w}$ \CAT states on $w$ qubits.

   Decomposing a $w$-qubit \CAT into 2 $w/2$-qubit \CAT states requires $w/2$ $ZZ$-measurements, each requiring 2 CNOTs, for a total of $w$. The next iteration requires 2 copies of $w/2$ CNOTs, also giving $w$. We require $\log w$ layers, for a total of $w\log w$ CNOTs (up to rounding due to $w$ not being divisable by $2$).
   Since we need $\frac{n}{w}$ many of them, this contributes $n\log w$ CNOTs.

   The total CNOT count is then $n\log w + n- 2w = n(\log w+1) -2w$. Replacing $w=t+1$ gives $n(\log (t+1)+1) - 2(t+1)$.

   For the depth we note that we only have to increase the depth until we get to \CAT states of size $2w$, which requires $\log 2w = \log w + 1$ layers. Each of these layers consist of $ZZ$-measurements, which have a CNOT depth of 2. However, since these CNOTs are to ancillae, that means that only half the data qubits are involved in a CNOT for the $ZZ$-measurement, so that the other half can already be used for the next layer of $ZZ$-measurements. Effectively that means that only the first and last $ZZ$ layer have a CNOT depth of 2, and all the intermediate layers can overlap so that they only contribute a CNOT depth of 1. The total CNOT depth is hence $\log w + 3 = \log (t+1) + 3 \leq \log t + 4$. If we don't use this compression technique to save on ancilla count, we get $2\log t + 2$ instead.

   Finally, for the ancilla count, note that each $ZZ$-measurement requires one ancilla, for a total of $\frac12 n(1+\log w) -w$. However, once we are done with a layer of $ZZ$-measurements, we measure that ancilla and it can be reused, so that at most $n/2$ ancillae are actually needed. If we use the CNOT layer compression described above, then two layers of $ZZ$-measurements must be kept active at the same time, doubling the ancilla count to $n$.
\end{proof}

\subsection{Proofs of \texorpdfstring{\cref{sec:circ-extract}}{Section 4}}

\begin{proof}[Proof of \cref{lem:unique-boundary}]
Suppose towards contradiction that there exists a boundary vertex connected to more than one boundary. When all three of its wires are outputs then it is its own connected component, which is impossible. Hence, it has exactly two output wires. Now consider an $X$-type fault on the edge labelled by $t$:
\[
    \scalebox{0.6}{\input{figs/flag-by-construction/decomp-ext.tikz}}
\]
This fault spreads to errors on the two neighbouring outputs.
Note that a GHZ state has only one $X$-type stabiliser generator and it acts on all qubits. Since we assume that we have at least 4 outputs, applying this $X$-stabiliser would flip the two errors to at least two other errors. So in both cases a weight-1 fault propagates to a fault of weight at least 2, contradicting fault-tolerance.
\end{proof}


\subsection{Proofs of \texorpdfstring{\cref{sec:reduction-to-z-spiders}}{Section 5.1}}
In this section, we show that any circuit implementing a \CAT state consisting of just CNOT gates, or CNOT+Hadamard gates gives rise to a `skeleton' of 3-ary $Z$-spiders which would allow us to construct a \CAT state using at least as many CNOT gates as using \cref{proc:circuit-extraction}. In \cref{subsubsec:opt-cnot}, we establish the minimality argument considering all possible CNOT circuits. In \cref{subsubsec:opt-cnot-h}, we generalise the argument to CNOT+Hadamard circuits. 

As was done in~\cite{rusch2025completeness}, we need to relax the notion of fault equivalence to an asymmetric relation called \emph{fault boundedness}. If two diagrams are mutually fault-bounded by each other, they are fault-equivalent.
\begin{definition}[$w$-fault boundedness, fault boundedness]\label{def:fault-boundedness}
 Let $D_1$ and $D_2$ be diagrams with noise models $\mathcal{N}_1$, respectively $\mathcal{N}_2$.
 We say that $D_1$ under $\mathcal{N}_1$ is \emph{$w$-fault-bounded} by $D_2$ under $\mathcal{N}_2$ for some $w \in \mathbb{N}^+$ if for any fault $F_1$ on $D_1$ with $\weightfunc_{\mathcal{N}_1}(F_1) < w$, we have either:
 \begin{itemize}
        \item $F_1$ is detectable, \textbf{or}
        \item there exists a fault $F_2$ on $D_2$ with $\weightfunc_{\mathcal{N}_2}(F_2) \leq \weightfunc_{\mathcal{N}_1}(F_1)$ such that $D_1^{F_1} = D_2^{F_2}$.
 \end{itemize}
 In this case we write $D_1 \wFaultBnd{w} D_2$, leaving the noise models implicit when they are clear from context.
 We say $D_1$ is \emph{fault-bounded} by $D_2$ if it is $w$-fault-bounded for all $w$. In this case we write $D_1 \FaultBnd D_2$.
\end{definition}
We interpret $D_1 \wFaultBnd{w} D_2$ as meaning that $D_1$ is `at least as' $w$-fault-tolerant as $D_2$, as any potentially problematic fault on $D_1$ would have a corresponding equally problematic fault on $D_2$.

\begin{lemma}
\[
    \scalebox{.8}{\input{figs/optimality/copy.tikz}}
\]
\label{lem:copy}
\end{lemma}   

\begin{proof}
Up to colour change, it suffices to show (a). Since any boundary edge flip on the \LHS corresponds to a boundary edge flip on \RHS of the same weight, fault-boundedness immediately follows.
\end{proof}    

\subsubsection{CNOT-Optimality over all CNOT Circuits}
\label{subsubsec:opt-cnot}

\begin{figure}[!htb]
    \[
        \scalebox{.8}{\input{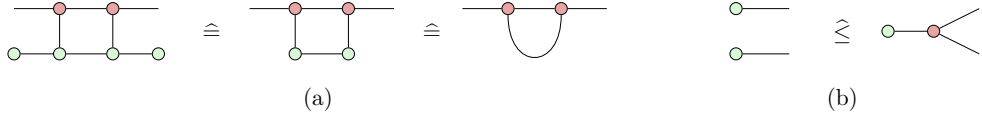}}
    \]
\caption{By $\textsc{Elim}_{\text{fe}}$ and $\textsc{Fuse-1}_{\text{fe}}$ in \cref{fig:fe-rewrites}, we can remove the 1-ary and 2-ary spiders in $(a)$ and get a spider skeleton composed of only one 3-ary spider. the \LHS of (b) is fault-bounded by the \RHS (see \cref{def:fault-boundedness}), this means the \LHS is `more' fault-tolerant. Hence, applying this rule to remove a 1-ary spider does not make the diagram less fault-tolerant.}
\label{fig:spider-skeleton}
\end{figure}
\begin{proof}[Proof of \cref{lem:CNOTs-to-spider-count}]
	Let $D$ be the associated diagram of $C$. Since $C$ $t$-fault-tolerantly implements $\ket{\CAT^n}$, $C\wFaultEq{t+1} \ket{\CAT^n}$. By \cref{prop:cat-state-noise-models}, $D\wFaultEq{t+1} \ket{\CAT^n}$. Note that reducing $D$ to $D'$ involves the application of spider fusion and identity removal rules. By~\cite{rodatz2025fault}, these are fault-equivalent rewrite rules. The diagrammatic reduction might also require the copy rule to remove 1-ary spiders. While this is not a fault-equivalent rewrite, it is true that applying it to remove spiders can only improve the fault-tolerance (see~\cref{lem:copy}). By~\cite{rodatz2025fault,rusch2025completeness}, fault equivalence and fault boundedness are both compositional and transitive. It follows that $D' \FaultBnd D \wFaultEq{t+1} \ket{\CAT^n}$. Now note that since the diagram of $\ket{\CAT^n}$ is a single spider with no internal wires, that any fault on it trivially corresponds to a fault on $D'$, so that we get the other direction of the fault-equivalence for free, so that we conclude $D' \wFaultEq{t+1} \ket{\CAT^n}$.
    
  Note that each 3-ary $Z$-spider in $D$ comes from a CNOT gate in $C$. Reducing $D$ to the spider-skeleton $D'$ can only remove 3-ary spiders by fusions or applications of the copy rule, but never introduces new 3-ary spiders. Hence, if $C$ contains $l$ CNOT gates, then $D$ contains $l$ 3-ary $Z$-spiders, and the skeleton $D'$ will contain $k\leq l$ 3-ary $Z$-spiders. We then see that if $D'$ contains $k$ 3-ary $Z$-spiders, that $C$ has at least $k$ CNOT gates, with at least one CNOT gate being `associated' to each $Z$-spider. 
    

	It remains to show that the number of CNOT gates is at least one more than the number of $Z$-spiders in $D'$. That is, $l\geq k+1$. To do this, we need to show that in the rewriting process from $D$ to $D'$, at least one 3-ary $Z$-spider gets removed.

	First, consider that in the circuit $C$, each spider occurs on a certain qubit, meaning that we can differentiate between when two spiders are connected on the same qubit line, or have their connection going between two different qubits. We say that spiders on the same qubit line are on the same \emph{path}. 



    Since $C$ implements a \CAT state, it has no free inputs and starts by preparing ancillae, and hence each path starts in some 1-ary spider, progressing along the path to either a measurement or an output wire. Suppose towards contradiction that no path in $D$ starts in a 1-ary $Z$-spider. This means they all start in a 1-ary $X$-spider. This means that all ancilla preparations in $C$ are in the $\ket{0}$ state. Hence, $C$ is a circuit that starts in the state $\ket{0\cdots 0}$, applies a CNOT circuit $U$, and then performs measurements in the $Z$- and $X$-bases. However, $\ket{0\cdots 0}$ is preserved by a CNOT circuit, so that $U\ket{0\cdots 0} = \ket{0\cdots 0}$, which means $C$ is equal to a product state, contradicting the fact that it is equal to a \CAT state. Hence, there must be at least one path in $D$ that starts in a 1-ary $Z$-spider.

    Consider a path of $D$ that starts in a 1-ary $Z$-spider. Since it is adjacent to a 3-ary spider, we proceed by case distinctions on this neighbouring spider. 
    \begin{description}
        \item[When it is an $X$-spider:]The copy rule applies, followed by the $\text{Elim}_{\text{fe}}$ and $\text{Fuse-1}_{\text{fe}}$. This removes at least one pair of 3-ary $Z$- and $X$-spiders in the skeleton, showing that the CNOT count of $C$ is at least one more than the number of 3-ary $Z$-spiders in the skeleton. This case is illustrated in \cref{fig:removal}.(a). 
        \item[When it is a $Z$-spider:]$\text{Elim}_{\text{fe}}$ and $\text{Fuse-1}_{\text{fe}}$ apply, removing at least one 3-ary $Z$-spider in the skeleton. Similar to the previous case, this shows that the CNOT count of $C$ is at least one more than the number of 3-ary $Z$-spiders in the skeleton. This case is illustrated in \cref{fig:removal}.(b).
    \end{description}  

    \begin{figure}[!htb]
    \[
        \scalebox{.8}{\input{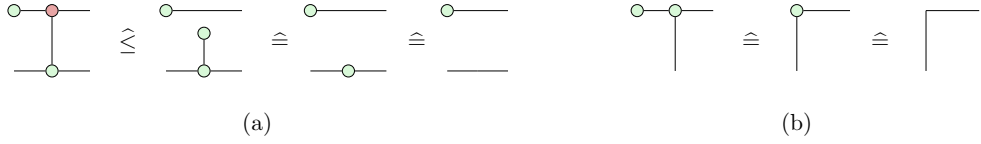}}
    \]
    \caption{At least one 3-ary $Z$-spider gets removed in the transition from $D$ to $D'$.}
    \label{fig:removal}
    \end{figure}

    Therefore, in both cases, the number of 3-ary spiders is strictly reduced when rewriting $D$ to $D'$. Hence, $l\geq k+1$. 
\end{proof}


\begin{proof}[Proof of \cref{prop:Zonly}]
First, we refer to \cref{fig:reduced-spider-skeleton} for a schematic representation of the construction of the reduced skeleton $S$.
Now, this proof consists of three steps: (1) first, we verify $S$ only contains 3-ary $Z$-spiders, (2) then we prove that $S$ implements the \CAT state and (3) finally, we show that $S$ fault-equivalently implements the \CAT state.

\textbf{$S$ only contains $Z$-spiders}
Let $D'$ be the spider-skeleton of $C$ and let $S$ be its reduced skeleton. Let $W'_X$ be the $X$-stabilising Pauli web of $D'$. Hence, $W'_X$ highlights all boundary edges of $D'$, as shown after step (1) of \cref{fig:reduced-spider-skeleton}. In $S$, we only keep the edges and spiders that were highlighted by $W'_X$. As $D'$ is phase-free, $W'_X$ may only contain highlighting in red (i.e.~Pauli $X$). All $X$-spiders in $D'$ must have an even number of edges highlighted in red ($X$), and since they are all 3-ary, this must mean that the resulting $X$-spiders in $S$ must either have degree zero or two. Thus, after removing the free-floating spiders and spiders of degree two, $S$ has no $X$-spiders. This is illustrated by \cref{fig:spiders}.(a) and (b). In addition, each $Z$-spider in $S$ comes from a $Z$-spider in $D'$, so that it has at most as many as $D'$ has.

\begin{figure}[!htb]
    \[
        \scalebox{.6}{\input{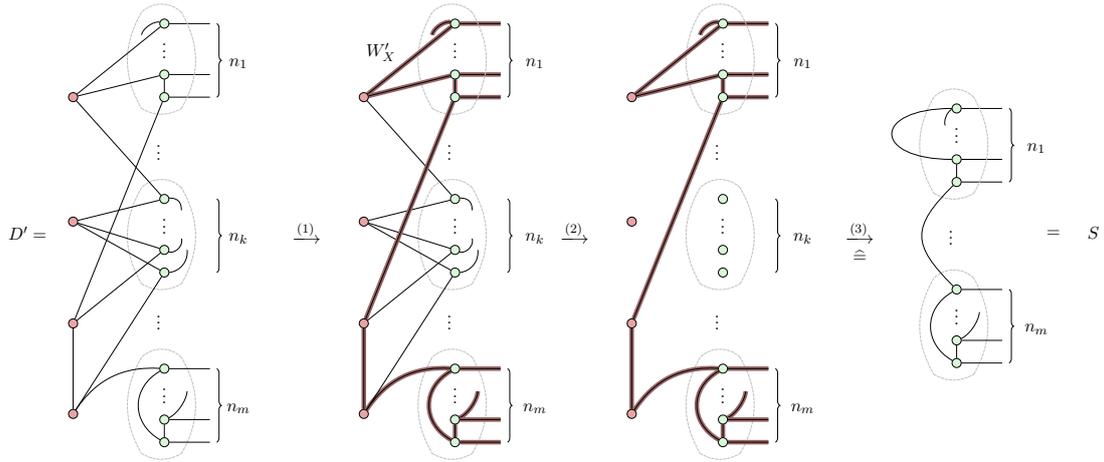}}
    \]
    \caption{Let $D'$ be a spider-skeleton of an $n$-qubit CAT state. By construction, $D'$ consists solely of 3-ary spiders. Using OCM, we rearrange the spiders in $D'$ into clusters of connected $X$- and $Z$-spiders. Each $X$-spider is incident to at most three $Z$-spiders, possibly belonging to different clusters. Each cluster of $Z$-spiders has $n_j \in \N$ output wires, with $\sum n_j = n$. In step (1) we fire the all-$X$-stabilising Pauli web $W'_X$ of $D'$. Note that some edges in $D'$ are not fired by $W'_X$. We drop these edges at step (2), yielding free-floating $X$- and $Z$-spiders. By removing these spiders and applying $\text{Elim}_{\text{fe}}$ to the diagram in step (3), we obtain a Z-graph which we call the \emph{reduced skeleton} $S$ of $D'$. This completes each step of \cref{def:reduced-spider-skeleton}.}
    \label{fig:reduced-spider-skeleton}
\end{figure}    

\begin{figure}[!htb]
    \[
        \scalebox{.65}{\input{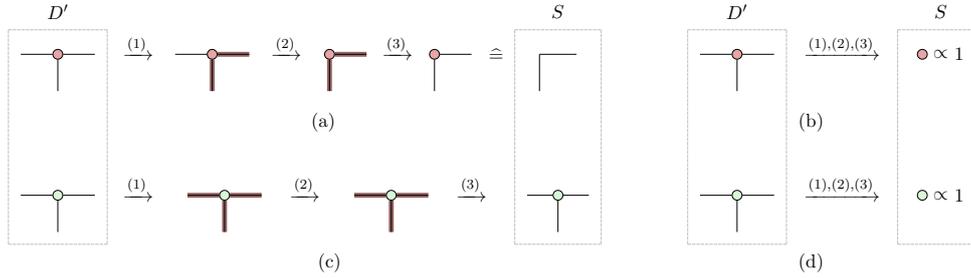}}
    \]
    \caption{Local transition of 3-ary spiders in \cref{fig:reduced-spider-skeleton}, with all cases considered. For an $X$-spider in $D'$, in cases (a) and (b), after removing edges not contained in $W'_X$ and applying $\text{Elim}_{\text{fe}}$ to the remaining 2-ary $X$-spiders, all $X$-spiders are eliminated from $D'$. Consequently, no $X$-spiders appear in $S$. For any $Z$-spider in $D'$, either all of its incident edges belong to $W'_X$ or none do. Accordingly, in cases (c) and (d), we either retain all edges of a $Z$-spider or discard them entirely. As a result, every $Z$-spider in $S$ has arity 3, and the total number of $Z$-spiders in $S$ is strictly smaller than in $D'$.}
    \label{fig:spiders}
\end{figure}

\textbf{$S$ is equal to the \CAT state}
To show that $S$ implements $\ket{\CAT^n}$, the $n$-qubit \CAT state, we will argue that it stabilises the same Pauli strings as $\ket{\CAT^n}$ by showing that it has corresponding stabilising Pauli webs.

First, we argue that $S$ is a connected Z-graph. Note that we picked $W'_X$ to be the smallest Pauli web stabilising the all-$X$ Pauli on the output. Hence, $W'_X$ does not contain any connected component not connected to an output, since in that case we could remove this part of the Pauli web to create a smaller Pauli web. Similarly, $W'_X$ does not consist of two unconnected components involving subsets of the free edges, since in that case the state would have two independent $X$-stabilisers, which is not true for the \CAT state. Hence $W'_X$ forms a single connected web, and so $S$ is also connected by construction.

By construction, $S$ stabilises all $X$, as is visualised in \cref{fig:z-stabilising}.(a). Namely, there exists an $X$-stabilising Pauli web $W_X$ in $S$ that highlights all boundary edges. Note that all spiders of $S$ are $Z$-spiders, and hence this Pauli web necessarily highlights all spiders since $S$ is connected. Hence, the condition that either all or none of their edges are highlighted in red ($X$) is satisfied, so that $W_X$ is valid Pauli web.

\begin{figure}[!htb]
    \[
        \scalebox{.5}{\input{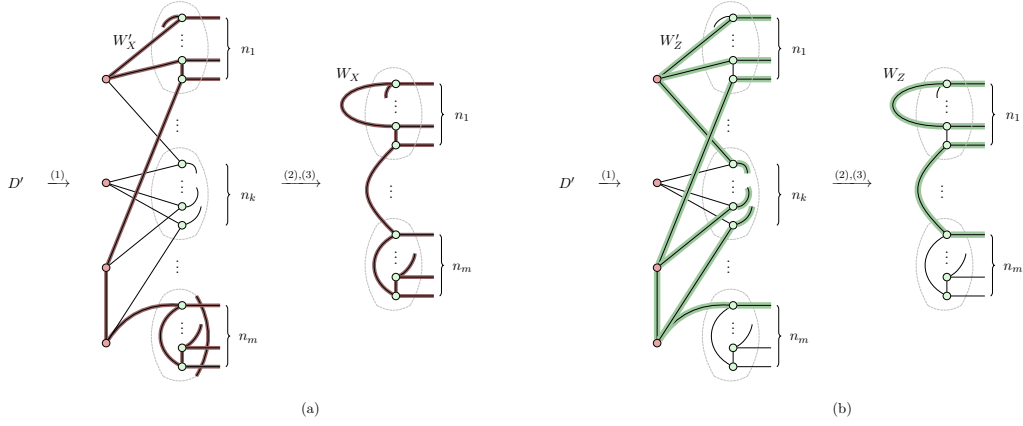}}
    \]
    \caption{Constructing the $X$- and $Z$-stabilising Pauli webs in the reduced skeleton $S$ based on the $X$- and $Z$-stabilising Pauli webs in the spider-skeleton $D'$.}
    \label{fig:z-stabilising}
\end{figure} 

Next, consider any Pauli string with evenly many Pauli $Z$'s (which is a stabiliser of $\ket{\CAT^n}$), as is shown in \cref{fig:z-stabilising}.(b). As $D'$ implements $\ket{\CAT^n}$, there exists a stabilising Pauli web $W'_Z$ on $D'$ that highlights the corresponding output edges. We will use $W'_Z$ to construct a $Z$-stabilising Pauli web $W_Z$ in $S$ that highlights the same output edges -- the edge highlighting $W_Z$ is obtained by tracking $W'_Z$ through the construction of $S$ from $D'$ in the natural way: for all edges that we remove from $D'$, drop the corresponding highlighting from $W'_Z$ in $S$. Furthermore, as we contract edges using $\text{Elim}_{\text{fe}}$, pick one of the highlightings of one of the two edges. For this to be unambiguously defined, we argue that both edges on either side of the $\text{Elim}_{\text{fe}}$ must have the same highlighting. As we only dropped edges that were not highlighted by the all-$X$ Pauli web, we either dropped none or all of the edges of $Z$-spiders. Therefore, the only spiders of degree two must be $X$-spiders. But then, in $W'_Z$, either all or none of the legs of any $X$-spider must be highlighted in green ($Z$). Hence, after one of the legs of the $X$-spider is removed, either both or none of the resulting wires are highlighted in green ($Z$). Therefore, $W_Z$ is uniquely defined. See \cref{fig:spiders3} for a visualisation of all possible cases.

\begin{figure}[!htb]
    \[
        \scalebox{.57}{\input{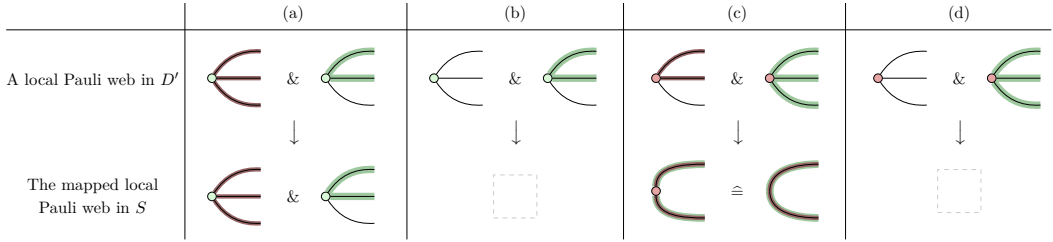}}
    \]
    \caption{Local transformations of a $Z$-stabilising Pauli web from $D'$ to $S$ induced by the transition of the $X$-stabilising Pauli web in \cref{fig:spiders,fig:reduced-spider-skeleton}. Let $D'$ be a spider-skeleton with the $X$-stabilising Pauli web $W'_X$ and a $Z$-stabilising Pauli web $W'_Z$. We consider a local neighbourhood of a degree-3 spider in $D'$ that is included in $W'_Z$. For a $Z$-spider, either all its edges or none of them are highlighted in $W'_X$; for an $X$-spider, either exactly two or none of the edges are. These possibilities correspond to cases (a),(b) and (c),(d), respectively. By \cref{def:reduced-spider-skeleton}, all incident edges in (a), and exactly two in (c), are retained in the reduced skeleton $S$, while all incident edges are removed in (b) and (d). Let $W_Z$ denote the image of $W'_Z$ in $S$. Accordingly, in $S$ the highlighted edges of the Z-spider are preserved in (a), all remaining edges of the $X$-spider are highlighted in (c), and no edges are highlighted in (b) nor (d). In each case, the local Pauli-web constraints are satisfied, so $W_Z$ remains valid.}
    \label{fig:spiders3}
\end{figure}  

Finally, we have to argue that $W_Z$ is a Pauli web. This is trivially the case, as $W'_Z$ respected the local conditions for Pauli webs so that for all edges in $S$, $W_Z$ must also respect the local highlighting conditions (the $Z$-spiders in $S$ come from $D'$ and hence must have evenly many edges highlighted). But then, we have shown that $S$ must stabilise any Pauli with evenly many $Z$'s. Therefore, $S$ indeed implements $\ket{\CAT^n}$. 


\textbf{Fault equivalence}
Finally, we argue that $S$ is $(t+1)$-fault-equivalent to $\ket{\CAT^n}$, the $n$-qubit \CAT state. Since $S$ only consists of $Z$-spiders, all $Z$-type edge flips can trivially be pushed to the boundary. Therefore, we are only concerned about $X$-type edges flips. Let $F$ be an undetectable fault on $S$ consisting of at most $t$ $X$-type faults. By~\cite{rusch2025completeness}, there exists some inconsequential fault $E$ such that $FE$ only acts on the boundary of $S$. We will now argue that $FE$ acts on at most $t$ boundary edges.

Take $F'$ on $D'$ to be the fault that acts on the same edges as $F$ on $S$. This is well-defined as each edge in $S$ comes from an edge in $D'$, or was the result of an application of $\text{Elim}_{\text{fe}}$.
As $F'$ is an all $X$-type fault, for edges that were newly obtained via $\text{Elim}_{\text{fe}}$ it does not matter which of the original edges in $D'$ we place the flip on (as these were $X$-spiders so that all edges are equivalent for $X$-type faults in $D'$). This is visualised in \cref{fig:fault-on-Dprime}. Note that $F'$ has the same weight as $F$.

\begin{figure}[!htb]
\[
    \scalebox{.6}{\input{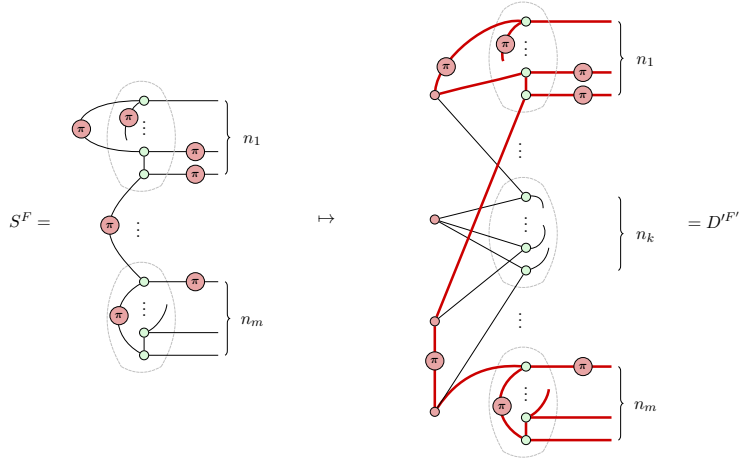}}
\]
\caption{Mapping an undetectable weight-$t$ fault $F$ on $S$ to a corresponding weight-$t$ fault $F'$ on $D'$. Highlighted edges in $D'$ indicate all edges in $S$ that are retained from \cref{def:reduced-spider-skeleton}.}
\label{fig:fault-on-Dprime}
\end{figure}


We can using the same mapping also construct a corresponding $E'$ on $D'$ from $E$. We claim that $F'E'$ also only acts on the boundary of $D'$. This is because if some $X$-faults in $S$ of $FE$ cancelled in $S$ because they surrounded a $Z$-spider, then this will also be true in $D'$, since it will be the `same' $Z$-spider. If instead two $X$-faults in $S$ cancel because they appear on the same edge, and this edge corresponded to an $X$-spider in $D'$, then they will also cancel, since there is an even number of them surrounding the $X$-spider. Hence, the internal faults of $F'E'$ also cancel in $D'$ and we are left with a boundary fault on the same output wires as in $S$. By~\cite{rusch2025completeness}, $F'$ must be an undetectable fault in $D'$, as it can be pushed to the boundary using $E'$. We note that $\operatorname{wt}(F') = \operatorname{wt}(F)\leq t$.


In \cref{lem:CNOTs-to-spider-count}, we show that $D'$ $t$-fault-tolerantly prepares $\ket{\CAT^n}$. This means $F'E'$ must act on at most $\operatorname{wt}(F')$ boundary edges (or on more than $n - \operatorname{wt}(F')$ boundary edges --- which by multiplication of the all-$X$ stabiliser is spacetime-equivalent to another fault that act on at most $\operatorname{wt}(F')$ boundary edges). Therefore, $FE$, which has equal weight to $F'E'$, also only acts on at most $\operatorname{wt}(F') = \operatorname{wt}(F)$ boundary edges (or on more than $n - \operatorname{wt}(F')$ boundary edges). This implies that every undetectable fault of weight $t$ in $S$ creates at most $t$ boundary edge flips and therefore $S$ $t$-fault-tolerantly prepares $\ket{\CAT^n}$.
\end{proof}

\begin{proof}[Proof of \cref{thm:CNOT-optimality}]
In \cref{lem:CNOTs-to-spider-count}, we show that if a CNOT circuit $C$ $t$-fault-tolerantly prepares $\ket{\CAT^n}$ and has $k$ CNOT gates, its spider-skeleton $D'$ contains at most $k - 1$ $Z$-spiders. According to \cref{def:associated-diagram,def:spider-skeleton,def:reduced-spider-skeleton,prop:Zonly}, we can efficiently find $S$, which is a Z-graph that is $t$-FE to $\ket{\CAT^n}$. Moreover, $S$ consists of at most $k - 1$ spiders.
\end{proof}  




\subsubsection{CNOT-Optimality over all CNOT+Hadamard Circuits}
\label{subsubsec:opt-cnot-h}

In this section we will generalise the ideas of the previous section which show how to reduce a generalised CNOT circuit that implements a fault-tolerant \CAT state to a Z-graph to work for a generalised CNOT+\emph{Hadamard} circuit. This will extend the proof of the CNOT optimality of our constructions based on Z-graphs to all circuits constructed using a combination of CNOT and Hadamard gates.



\begin{definition}
	A \emph{generalised \CH circuit} consists of CNOT gates, Hadamard gates, ancillae preparations in the $\ket{0}$ and $\ket{+}$ states, measurements in the $Z$- or $X$-basis, and Pauli $Z$ or $X$ corrections based on these measurement outcomes satisfying the constraints listed in \cref{def:gen-cnot}.
\end{definition}

In \cref{fig:hadamard}, (a) represents a Hadmard gate in ZX-calculus. (b) demonstrates how the highlighting of a Pauli web is flipped when it passes through a Hadamard gate. (c) and (d) introduces two fault-equivalent rewrite rules that will be used in the remainder of this section~\cite{rodatz2025fault,rusch2025completeness}.

\begin{figure}[!htb]
\[
    \scalebox{.8}{\input{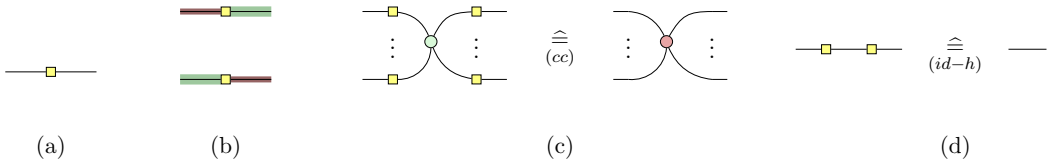}}
\]
\caption{The representation of a Hadamard in a ZX-diagram is shown in (a). A Hadamard sends Pauli $X$ to $Z$ and $Z$ to $X$ via conjugation. Correspondingly, (b) visualises its Pauli web. As shown in (c) and (d), the colour-change rewrite rule (cc) and the Hadamard cancellation rule (id-h) are fault-equivalent.}
\label{fig:hadamard}
\end{figure}

\begin{procedure}
A generalised CNOT+Hadamard circuit $C_{\CH}$ can be represented by a ZX-diagram by translating the circuit post-selected onto the all 0 measurement outcome into a diagram following the table of \cref{fig:mappings} and \cref{fig:hadamard}.(a). We refer to this as the \emph{associated diagram} $D_{\CH}$. 
\label{proc:associated-cnot-h}
\end{procedure}    

\begin{remark}
As shown in \cref{fig:component-cnot-h}, $D_{\CH}$ is composed of 3- and 1-ary spiders as well as Hadamard gates. It is equal to the linear map that $C_{\CH}$ deterministically implements. Reasoning analogously as \cref{prop:cat-state-noise-models}, we can show that $D_{\CH} \;\FaultEq\; C_{\CH}$.
\label{rmk:associated-cnot-h}
\end{remark}

\begin{figure}[!htb]
\[
    \scalebox{.8}{\input{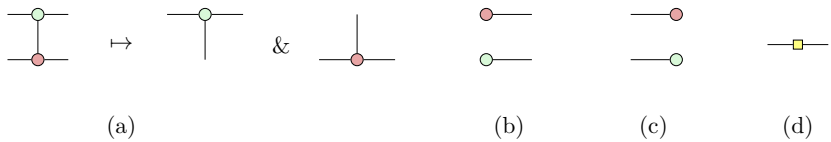}}
\]
\caption{$D_{\CH}$ contains 3-ary $Z$- and $X$-spiders (a, corresponding to CNOT gates), 1-ary spiders (b and c, corresponding to ancilla preparations and postselected measurement outcomes), and Hadamard gates (d).}
\label{fig:component-cnot-h}
\end{figure}


To show that Z-graphs can continue to represent optimal \CAT-state preparations even under the larger class of \CH circuits we have to show that the addition of Hadamard gates doesn't allow us to get a lower CNOT count. We do this by showing that we can reduce the relevant part of the circuit to be Hadamard-free, which allows us to then use the techniques from the previous section.

As in the construction of the reduced skeleton in the previous section, we will do this by considering the Pauli web $W'_X$ of $D_{\CH}$ that highlights all the free edges red ($X$), which must exist since $D_{\CH}$ implements the \CAT state. Since the diagram now contains Hadamards, this Pauli web is now not necessarily coloured fully red ($X$). Instead, each Hadamard the web `encounters' changes its colour to green ($Z$). We will use the fault-equivalent rewrites of \cref{fig:hadamard} to rewrite the circuit to make this whole Pauli web red ($X$).

\begin{procedure}\label{proc:reduce-green}
  Let $D_{\CH}$ be the associated diagram of a \CH circuit, and let $W'_X$ be the smallest Pauli web highlighting all free edges red ($X$). We produce the \emph{H-reduced circuit} $D'_{\CH}$ in the following way. Keep doing the following steps as long as at least one Hadamard is highlighted by $W'_X$.
  \begin{equation*}
    \input{figs/optimality/hadamard-pauli-reduce-1.tikz}
  \end{equation*}
  Each step reduces the number of edges or edge fragments coloured green ($Z$) and hence this terminates. Note that as long as one Hadamard has an adjacent edge coloured, that it will always have a red ($X$) side and a green ($Z$) side, and that the only way that $W'_X$ can contain green ($Z$) edges is if it encounters a Hadamard along its path since it starts fully red ($X$) on the outputs. Hence, after no Hadamard is coloured, the entire Pauli web $W'_X$ must be coloured red ($X$).
\end{procedure}
Note that the H-reduced circuit $D'_{\CH}$ can contain a combination of CNOT and CZ gates, since we are pushing the Hadamards freely. However, the CZ gates will only appear on the part of the diagram not coloured red ($X$) by the Pauli web.

The reason we do this transformation to the H-reduced circuit, is because each 3-ary Z-spider in $D'_{\CH}$ that is fully coloured red ($X$) by $W'_X$ we can now associate a CNOT gate, giving us a way to relate the CNOT cost of the original circuit to the number of Z-spiders. Note that \cref{proc:reduce-green} didn't change the number of entangling gates in the circuit.

We can now apply the rewrites of \cref{fig:spider-skeleton} again to remove the 1-ary spiders from the diagram to build the skeleton of this H-reduced circuit. After we do this, the part of $D'_{\CH}$ that is coloured red ($X$) by $W'_X$ consists solely of 3-ary Z-spiders and X-spiders and does not contain any Hadamards. We can then apply \cref{def:reduced-spider-skeleton} to produce the reduced skeleton in exactly the same way, only keeping the part of the diagram that is coloured red ($X$).
Proving that the resulting diagram is a fault-tolerant cat state follows entirely analogously as that of \cref{prop:Zonly}. This is because the complications with the Hadamards only appear outside the region coloured red ($X$), and that entire region gets dropped in the reduced skeleton, so that we don't have to consider faults there, or what happens if a Z-Pauli web gets turned red ($X$) by a Hadamard there. We then have the analogous version of \cref{thm:CNOT-optimality}, but then for \CH circuits.

\begin{theorem}\label{thm:CNOT-H-optimality}
    Let $C$ be generalised \CH circuit containing $k$ CNOT gates, which $t$-fault-tolerantly prepares the $n$-qubit \CAT state.
  Then we can efficiently find a $Z$-graph containing at most $k-1$ spiders, which also $t$-fault-tolerantly prepares the $n$-qubit \CAT state.
\end{theorem}

\subsection{Proofs of \texorpdfstring{\cref{sec:ft-cat-graph}}{Section 5.2}}
\label{subsec:ft-construction}

\begin{proof}[Proof of \cref{prop:undetectable-fault-is-cut}]
($\Rightarrow$)
Let $F$ be an undetectable fault on $D$.
If $D$ implements the \CAT state fault tolerantly, there exists an equivalent fault $F'$ that acts only on the boundary of $D$.
Since $F$ and $F'$ are equivalent, there exists a set of spider stabilisers $\mathcal S$ such that
$F \prod_{s \in \mathcal S} s = F'$~\cite{rusch2025completeness}.
Because $D$ consists only of $Z$ spiders and $F$ consists only of $X$ spiders, the stabilisers in $\mathcal S$ may be chosen to be $Z$ spiders fired in the $X$ basis.
Thus, $V_{\mathcal S} \subseteq V$ the set of spiders fired by $\mathcal S$ determines a bipartition $(V_{\mathcal S}, V \setminus V_{\mathcal S})$ of the vertex set of $G$.

Firing a $Z$ spider in the $X$ basis toggles the presence of an $X$ fault on each incident edge.
Consequently, after firing all spiders in $V_{\mathcal S}$, an edge carries an $X$ fault if and only if it is incident to exactly one vertex in $V_{\mathcal S}$.
Since $F'$ acts only on the boundary, every internal edge must be toggled an even number of times, and hence the set of faulty edges is precisely the set of edges crossing the cut $(V_{\mathcal S}, V \setminus V_{\mathcal S})$.
Therefore, $F$ corresponds to a cut of $G$.

($\Leftarrow$)
Conversely, suppose $F$ corresponds to a cut of $G$.
Then there exists a partition $V = V_1 \sqcup V_2$ such that the faulty edges are exactly those with one endpoint in $V_1$ and the other in $V_2$.
As before, firing all spiders in $V_1$ in the $X$ basis removes all faults on internal edges and pushes the fault entirely to the boundary.
Since the fault $F$ can be pushed to the boundary, it is therefore undetectable.
\end{proof}

\begin{proof}[Proof of \cref{prop:cnot-optimal}]
Let $C$ be a generalised CNOT circuit that $t$-fault-tolerantly prepares $\ket{\CAT^n}$. Let $S$ be the reduced skeleton of $C$, then $S$ is a $Z$-graph. Let $(G,M)$ be its underlying marked graph. Suppose that $G$ has $v$ vertices and $M$ contains $n$ marks. Let $r_t^n$ be the minimum vertex ratio over all $t$-robust graphs with $n$ marks.

Suppose towards contradiction that the CNOT count of $C$ is at most $n(r_t^n +1)$. By \cref{thm:CNOT-optimality}, we can efficiently find a $Z$-graph $S$ containing at most $n(r_t^n +1) - 1$ spiders, where $S$ is the reduced skeleton of $C$ and it $t$-fault-tolerantly prepares $\ket{\CAT^n}$. This implies the number of vertices in $S$ is at most $nr_t^n - 1$, which means the vertex ratio is at most $(nr_t^n - 1)/n = r_t^n - 1/n < r_t^n$. This yields a contradiction. Therefore, any generalised CNOT circuit that $t$-fault-tolerantly prepares the $n$-qubit \CAT state contains at least $n(r_t^n+1)+1$ CNOT gates. 
\end{proof}




\begin{definition}
    A \emph{unit tree} $c_t$ is the elementary local pattern that we need to consider when decomposing a phase-free $Z$-diagram up to a weight-$t$ fault under the edge-flip noise model. For simplicity, we use vertices in place of all internal spiders. Each unit tree is three-regular. When an edge has two vertices, we call it an \emph{internal edge}. Otherwise, we call it a \emph{leaf}. The number of leaves in $c_t$ is $t$.
    \begin{figure}[!htb]
        \[
            \scalebox{1}{\input{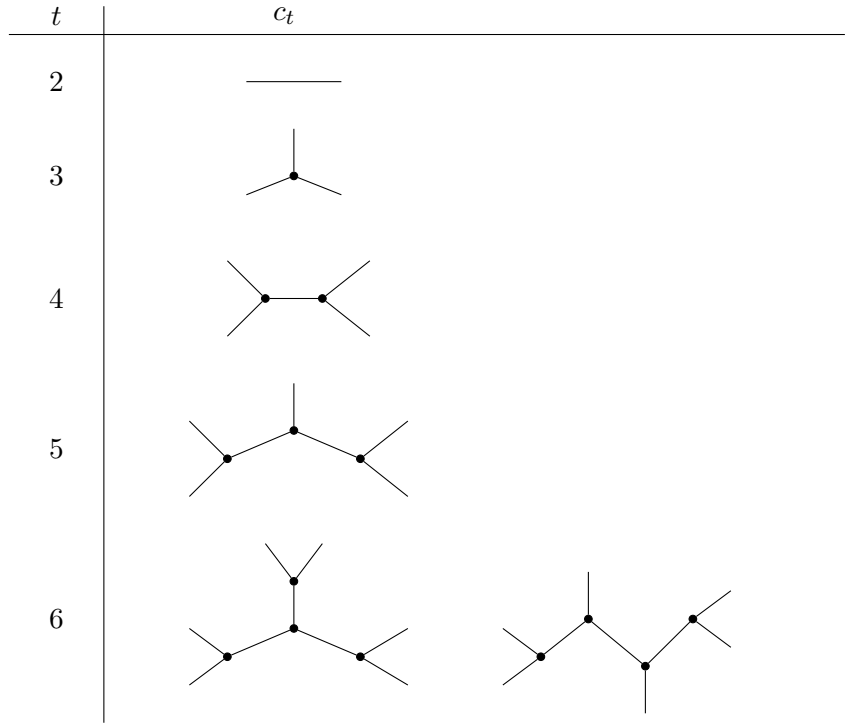}}
        \]
        \caption{The unit trees for $2\leq t \leq 6$. Each unit tree is a local patterns that contains no cycles and has $t$ leaves. In particular, it can contain at most $t$ marks.}
        \label{fig:unit-branch}
    \end{figure}    
\label{def:unit-cells}
\end{definition}  

\begin{lemma}
For $t=3$, the optimal vertex ratio is $r_3 = \frac23$.
\label{lem:vertex-ratio-r3}
\end{lemma} 
\begin{proof}
Since $t=3$, it suffices to consider three fault locations. This leads us to consider the unit tree $c_3$ in \cref{fig:unit-branch}. Firstly, we show that $r_3 \geq \frac23$. Suppose towards contradiction that $r_3 < \frac23$.
Let $v$ and $n$ be the number of vertices and target qubits respectively in a graph with vertex ratio $\frac{v}{n} < \frac{2}{3}$. Let $m_3(v)$ be the number of marks on edges neighbouring vertex $v$. We have the inequality on the average number of marks around a vertex
\begin{equation}
    \frac{\sum_{v\in V}m_3(v)}{|V|}=\frac{2n}{v}>3.
\end{equation}
Here we have used that every mark gets counted twice, once by each of its neighbouring vertices, and that $r_3=\frac{v}{n}$. By the pigeonhole principle, some vertex $v$ must have $m_3(v)\geq4$, in which case we get a problematic fault.

Next, we show that $r_3 \leq \frac23$ by giving an explicit construction of a family of graphs. In \eqref{eq:c3-construction}, the graph has no weight-3 non-local cuts.
\begin{equation}
    \scalebox{1}{\input{figs/flag-by-construction/c3-construction.tikz}}
\label{eq:c3-construction}    
\end{equation}
Thus, we only have to consider local weight-$3$ faults.

By adding one mark on each edge we have the following.
\[
\scalebox{1}{\input{figs/flag-by-construction/c3-construction2.tikz}}
\]
Since the only undetectable weight-3 faults must act around the same vertex, we see that this encloses precisely 3 marks, meaning we can push the fault to the boundary where it becomes a weight-3 fault.

This proves $r_3 = \frac23$.
\end{proof}    

For the following results, note that even though there will be fewer marks than edges, in principle it is still possible for there to be multiple marks on the same edge. This number can't be higher than 2 as then a weight 2 fault would propagate to a higher-weight fault. Hence, the number of marks per edge is either 0, 1 or 2.

\begin{lemma}
For $t=4$, the optimal vertex ratio is $r_4 = \frac56$.
\label{lem:vertex-ratio-r4}
\end{lemma}
\begin{proof}
We first show that $r_4\geq\frac56$.
Every edge neighbours 4 other edges.
For each edge, we will count the number of marks on it or on neighbouring edges.
Let $m_4(e)$ be the number of marks on an edge $e$ or on any of the four edges neighbouring edge $e$. For the purposes of contradiction, assume that $r_4<\frac56$ and let $n$, $v$, and $e$ be the number of marks, vertices, and edges in a graph with a vertex ratio $\frac vn<\frac56$. We have the inequality on the average number of marks around an edge
\begin{equation}
    \frac{\sum_{e\in E}m_4(e)}{|E|}=\frac{5n}{e}=\frac{5n}{\frac32v}>4.
\end{equation}
Here we have used that every mark gets counted 5 times, once by the edge it is on and once by each of its edge's neighbouring edges. By the pigeonhole principle, some vertex $e$ must have $m_4(e)\geq5$, in which case we get a problematic fault.
Since a cut of weight 4 on the four edges around edge $e$ will push out to 5 marks, this would not be fault-tolerant, and thus our assumption that $r_4<\frac56$ must be wrong.

We now show that $r_4\leq\frac56$.
\begin{figure}
    \centering
    \scalebox{1}{\input{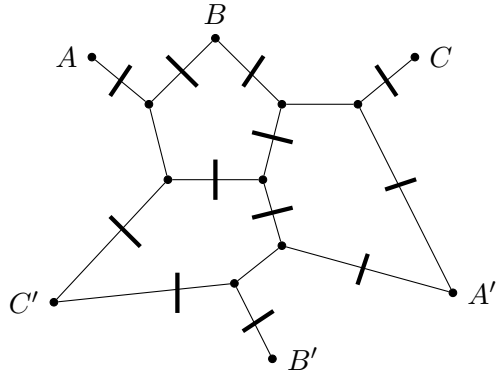}}
    \caption{Shown is a fragment of a graph with 10 vertices, 15 edges, and 12 marks. The vertex ratio is the optimal amount for $t=4$ of $\frac56$. The vertices $A$, $C$, and $B'$ are only $\frac13$ of a vertex each and $B$, $A'$, and $C'$ are $\frac23$ of a vertex each. Any $k$ copies of this graph fragment can be put together into a larger graph where $A$, $B$, and $C$ in one copy should be connected to $A'$, $B'$, and $C'$ in the next. When $k=1$ this produces the Petersen graph and when $k=2$ this produces the dodecahedron graph.}
    \label{fig:r4-optimal}
\end{figure}

In Figure~\ref{fig:r4-optimal} we see an example of a graph fragment with the optimal vertex ratio of $\frac56$.
Copies of this graph fragment can be identified together to create a 3-regular graph with $10k$ vertices for any positive integer $k$.
For any tree of perimeter 4, as seen in Figure~\ref{fig:unit-branch}, the graph fragment contains at most 4 marks, even across the boundary.
This is guaranteed by having at most 1 mark per edge and no two adjacent vertices surrounded by 3 marks.
We also note that no number of copies of this graph can be cut by a nonlocal cut of weight 4 or less.

Thus we have $r_4=\frac56$.
\end{proof}

\begin{lemma}
For $t=5$, the optimal vertex ratio is $r_5 = 1$.
\label{lem:vertex-ratio-r5}
\end{lemma}
\begin{proof}
We first show that $r_5\geq1$.
We will argue by examining a generic fragment of a 3-regular graph.
Note that since we have the condition that our graph will not admit any cuts of weight 5 or less that cut off any cycles, we know the length of the shortest cycle in the graph is at least 6.

We will frequently refer to the number of marks `around' a vertex, by which we mean marks on the three edges incident to a given vertex. Any mark is around two vertices.
Note that no vertex can have more than 3 marks around it as otherwise a weight 3 fault would cut off more than 3 marks, and thus push out to more than 3 free edges.

\begin{figure}
    \centering
    \scalebox{1}{\input{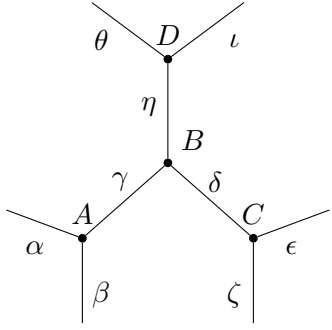}}
    \caption{A generic fragment of a 3-regular graph. Vertices are labeled $A-D$ and edges are labeled with Greek letters.}
    \label{fig:r5-tree}
\end{figure}

In this argument, we make heavy reference to Figure~\ref{fig:r5-tree} as we will need to consider several cases on this diagram.
We will prove that for every vertex with 3 marks around it, we can show that at least one vertex with no more than 1 mark around it can exist.
This will show that the average number of marks around a vertex is no greater than 2, and since each mark is counted twice, this means the number of marks does not exceed the number of vertices, which would show that the vertex ratio is $\geq1$.

In Figure~\ref{fig:r5-tree}, suppose that vertex $A$ has 3 marks around it.

We first show that neither $C$ nor $D$ can have 3 marks around it.
We start with the latter claim. If $C$ (resp. $D$) has 3 marks around it, we can cut $\{\alpha,\beta,\epsilon,\zeta,\eta\}$ (resp. $\{\alpha,\beta,\delta,\theta,\iota\}$), meaning a weight-5 fault propagates to a weight-6 fault on the boundary, a contradiction.

Now we will show that this implies that at least one of $B$, $C$, or $D$ has at most 1 mark around it.
Having done this, we will `group' $A$ with this 0-mark or 1-mark vertex along edge $\gamma$, and group that vertex with $A$ along its edge pointing towards $A$, which is $\gamma$, $\delta$, or $\eta$ for $B$, $C$, and $D$, respectively.

We first consider a specific case, which will be important later.
If $\gamma$ has 0 marks and $B$ has 1 mark around it, we know that it must be on $\delta$ or $\eta$. If it is on $\eta$, then $C$ must have at most one mark around it, as otherwise cutting $\{\alpha,\beta,\epsilon,\zeta,\eta\}$ would cut off at least 6 marks, which is impossible. In this case, group $C$ with $A$.

If $\gamma$ has 1 mark and it is the only mark around $B$, we group $B$ with $A$.
Otherwise, $B$ has at least 2 marks and we know that one of $\delta$ and $\eta$ contains 0 marks, as otherwise cutting $\{\alpha,\beta,\delta,\eta\}$ cuts off at least 5 marks.
Suppose $\delta$ has 0 marks and $\eta$ has more than 0.
Then, either $C$ has at most 1 mark around it or cutting $\{\alpha,\beta,\epsilon,\zeta,\eta\}$ cuts off at least 6 marks, which is impossible.
The case where $\eta$ has 0 marks and $\delta$ has at least one is analogous, except that we end up grouping $D$ with $A$.
If both $\delta$ and $\eta$ have 0 marks, then either one of $C$ and $D$ has 1 or fewer marks and can be grouped with $A$ or cutting $\{\gamma,\epsilon,\zeta,\theta,\iota\}$ cuts off at least 6 marks, which is impossible.

Note that we have found at least one vertex with at most 1 mark and grouped it with $A$ while exploring the graph along the edge $\gamma$ from $A$.
We now show that every vertex with at most 1 mark can be `grouped' to at most one vertex with 3 marks along each of its edges.
As we have already shown that $C$ and $D$ cannot have 3 marks around them, we know that no two vertices with 3 marks can be at a distance of 2 from each other.
We have only left to show that a vertex with at most one mark cannot be grouped along the same edge with two vertices with 3 marks, one at distance 1 and one at distance 2.
As a concrete example, this would require showing that the case where $A$ and $B$ have 3 marks around them and have both been grouped with $C$, where $C$ has 1 mark around it, is impossible.
Having 3 marks around both $A$ and $B$ requires exactly 2 marks on $\gamma$.
The only case where we grouped a vertex with 3 marks around it with an adjacent vertex required exactly one mark on their connecting edge.
Thus, if $B$ was grouped with $C$, then $\delta$ would contain a mark.
The procedure above would result in $D$ being grouped with $A$, not $C$.

This shows that along each edge coming out of a vertex with 3 marks around it, a unique edge of a vertex with 1 mark around it can be identified.
Thus, the total number of vertices with 1 mark around them cannot be smaller than the total number of vertices with 3 marks. This shows that the average number of marks around a vertex will be at most 2, which completes the proof that $r_5\geq1$.

We now show that $r_5\leq1$.
\begin{figure}
    \centering
    \scalebox{1}{\input{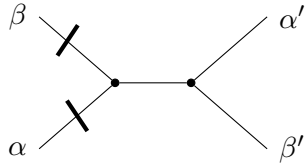}}
    \caption{Shown is a fragment of a graph with 2 vertices, 4 half-edges and 1 full edge, and 2 marks. The vertex ratio is the optimal amount for $t=5$ of $1$. The graph can be tiled in a square grid with many choices for how to wrap the boundary conditions. The half-edge $\alpha$ in one tile should be glued to $\alpha'$ in the next, and analogously for $\beta$ and $\beta'$.}
    \label{fig:r5-patch}
\end{figure}

In Figure~\ref{fig:r5-patch} we see a patch of a graph with 2 vertices, 4 half-edges and 1 full edge, and 2 marks.
This means a graph made of such patches will have a vertex ratio of $1$.
Furthermore, we can tile these patches into a square grid to get an infinite family of graphs of this vertex ratio.
Such a grid made of these patches will produce a square tiling, which we can then place on a torus to fix boundary conditions.
This will not admit any nonlocal cuts of weight 5 or less and all size 5 unit trees on this graph have at most as many marks as leaves.

This proves $r_5=1$.
\end{proof}

We now consider the case of $t=\infty$.
Using some tools from the theory of graph expansion, we can find a decomposition of a \CAT into linearly many degree-3 spiders.
As we will see below, this construction relies on the constructions of a family of expander graphs known as Ramanujan graphs.

\begin{definition}
\label{def:ramanujan}
    Let $G=(V,E)$ be a connected $d$-regular simple graph where $\lambda_i$ are the eigenvalues of its adjacency matrix sorted as $d=\lambda_1\geq\lambda_2\geq\dots\geq\lambda_{|V|}\geq -d$. The graph $G$ is called a \textit{Ramanujan graph} if $\max(\left|\lambda_2\right|,\left|\lambda_3\right|,\dots,\left|\lambda_{|V|}\right|)\leq2\sqrt{d-1}$.
\end{definition}

\begin{definition}
\label{def:cheeger}
    Let $G$ be a graph with vertices $V$ and edges $E$. We call the \textit{Cheeger constant} $h(G)$ of the graph $G$ as the smallest ratio of the number of edges separating a group of vertices. Formally,
    \begin{equation}
        h(G)=\min\limits_{\substack{A\subseteq V\\|A|\leq\frac12|V|}}\frac{|\partial A|}{|A|}
    \end{equation} where $\partial A$ denotes the number of edges connecting a vertex in $A$ to a vertex not in $A$.
\end{definition}

\begin{theorem}[Theorem 2.4 in~\cite{hoory2006expander}]
\label{thm:cheeger}
    Let $G=(V,E)$ be a connected $d$-regular simple graph where $\lambda_i$ are the eigenvalues of its adjacency matrix sorted as $d=\lambda_1\geq\lambda_2\geq\dots\geq\lambda_{|V|}\geq -d$. Then
    \begin{equation}
        \frac{d-\lambda_2}{2}\leq h(G)\leq\sqrt{2d(d-\lambda_2)}.
    \end{equation}
\end{theorem}

\begin{corollary}
\label{cor:ramanujan_cheeger}
    If $G$ is a $d$-regular Ramanujan graph, we have
    \begin{equation}
        \frac{d}{2}-\sqrt{d-1}\leq h(G).
    \end{equation}
\end{corollary}
\begin{proof}
    This is an immediate consequence of Definition~\ref{def:ramanujan} and Theorem~\ref{thm:cheeger}
\end{proof}

By using Corollary~\ref{cor:ramanujan_cheeger}, we can consider faults on a $d$-regular Ramanujan graph $G$.
We know that any undetectable fault of weight $w$ will cut the graph and the smaller resulting component will contain at most $\frac{w}{h(G)}$ vertices.
This means that if we add at most $h(G)$ free edges to each vertex of $G$ and add no other marks on the edges of $G$, then no undetectable fault can be pushed out to a higher-weight fault.

\begin{theorem}[Morgenstern~\cite{morgenstern1994existence}]
\label{thm:ramanujan_construction}
There exist explicit constructions of regular Ramanujan graphs. Such graphs have $q^{6k}-q^{2k}$ vertices each of degree $q+1$ where $k$ is a positive integer and $q$ is a prime power.
\end{theorem}

\begin{theorem}
\label{thm:linear_cat_graph}
    There exist fault-equivalent decompositions of \CAT states into three-legged spiders where the \CAT state has $\lfloor\frac{q+1}2-\sqrt{q}\rfloor(q^{6k}-q^{2k})$ qubits, where $k$ is a positive integer and $q$ is a prime power.
    The number of three-legged spiders used in this decomposition is equal to $(q^{6k}-q^{2k})\times c$, where $c$ is the number of three-legged spiders that it takes to fault-equivalently decompose a \CAT state with $\lfloor\frac{3q+3}2-\sqrt{q}\rfloor$ qubits.
\end{theorem}
\begin{proof}
    By Theorem~\ref{thm:ramanujan_construction}, we first construct any $(q+1)$-regular Ramanujan graph with $q^{6k}-q^{2k}$ vertices.
    Next, we replace each vertex with a $Z$-spider and we add $\lfloor\frac{q+1}2-\sqrt{q}\rfloor$ free edges, meaning now each vertex is of degree $\lfloor\frac{3q+3}2-\sqrt{q}\rfloor$.
    We note that any atomic fault on the edges of this graph will either be on a free edge, in which case we can find an analogous fault on a free edge of the \CAT state, or it will be on an internal edge.
    Any $Z$ component of a fault can be pushed out to at most one free edge, which might also contain an $X$ if the fault is not a pure $Z$ fault.
    Thus, we merely consider internal $X$ faults.
    
    It suffices to consider $X$ faults on a set of edges that cut the graph.
    If the $X$ faults have total weight $w$, then we know by Definition~\ref{def:cheeger} that the smaller side of the cut will contain at most $\frac{w}{h(G)}$ vertices.
    Since each vertex is connected to $\lfloor\frac{q+1}2-\sqrt{q}\rfloor$ free edges and by Corollary~\ref{cor:ramanujan_cheeger}, $\frac{q+1}2-\sqrt{q}\leq h(G)$, we have that the smaller side of the cut will contain at most $w$ free edges, meaning the fault will be pushed out to fault of weight at most $w$ around a \CAT state.

    Having shown that this construction is fault-equivalent to a \CAT state, we now fault-equivalently rewrite the $(\lfloor\frac{3q+3}2-\sqrt{q}\rfloor)$-legged $Z$ spider at each vertex into $c$ three-legged spiders, completing the proof.
\end{proof}

\begin{theorem}
For $t=\infty$, we have $r_\infty\coloneqq \liminf\limits_{n\to\infty}\sup_tr_t^n \leq 12$.
\label{thm:vertex-ratio-infinity}
\end{theorem}
\begin{proof}
    By Theorem~\ref{thm:linear_cat_graph}, we see that a graph with the right number of vertices will have $c$ three-legged spiders for every $\lfloor\frac{q+1}2-\sqrt{q}\rfloor$ free edges, where $c$ is the number of three-legged spiders that it takes to fault-equivalently decompose a \CAT state with $\lfloor\frac{3q+3}2-\sqrt{q}\rfloor$ qubits. Of these, $\lfloor\frac{q+1}2-\sqrt{q}\rfloor$ vertices will be connected to free edges and the remaining $c-\lfloor\frac{q+1}2-\sqrt{q}\rfloor$ will form a 3-regular graph.

    It suffices to choose an appropriate value for $q$, where $q$ is a prime power.
    If we choose $q=9$, we get $\lfloor\frac{q+1}2-\sqrt{q}\rfloor=2$, so every vertex of degree 10 gets an additional 2 free edges, turning it into a 12-legged spider.
    A 12-qubit \CAT state can be fault-tolerantly decomposed into a graph with 14 internal vertices and 12 boundary vertices.
    Of these 26 vertices, 2 correspond to free edges added to the vertices in the Ramanujan graph, leaving us with a rewrite containing 24 internal three-legged spiders at the location of each vertex in the Ramanujan graph. This proves that for particular values of $n$, where $n=2(9^{6k}-9^{2k})$ and $k$ is a positive integer, we can fault-equivalently rewrite an $n$-qubit \CAT state with a vertex ratio of at most $\frac{24}2=12$.
\end{proof}

\begin{proof}[Proof of \cref{cor:constant-CNOT-overhead}]
  Using the above construction, we note that degree-$(q+1)$ Ramanujan graphs exist with $2(q^{6k}-q^{2k})$ vertices where $q$ is a prime power and $k$ is a natural number. We also observe the fact that the ratio between these values for some $k$ and $k+1$ will never be larger than $9^6+9^2=531522$.
  This implies that if we wish to construct an $n$-qubit \CAT state for any $n$, we can `round up' $n$ to the nearest value of the form $2(9^{6k}-9^{2k})$, which will be an increase of at most a constant factor.
  By our construction, the number of degree-10 vertices in the Ramanujan graph is proportional to the number of three-legged spiders required to express it.
  To get the best possible factor, we can use an 8-regular Ramanujan graph, where the rounding up factor will be bounded by $7^6+7^2=117698$, where one output is added to each vertex, rewriting it into a 9-qubit \CAT state expressed using 14 vertices and 1 mark, for a worst-case vertex ratio bound of 1647772, a constant!
  This proves any \CAT state can be constructed fault-tolerantly in $O(n)$ resources such as gates, flags, and ancillae.

  However, since nearly all sufficiently large graphs are expanders, the true value of $r_\infty^n$ is likely to be quite close to $r_\infty$ for nearly all $n$, so any large constants merely provide provable worst-case values. 
\end{proof}

\begin{proof}[Proof of \cref{thm:depth-4}]
Our strategy will be to turn every internal and boundary vertex into a 3-qubit \CAT state.
Notice that this means we are unpacking the marks into $n$ degree-3 vertices.
We will then be able to `glue' together output edges of these vertices using Bell basis measurements.
It takes 2 CNOTs to make a 3-qubit cat state and we need 1 more CNOT per measurement.
Any output edge does not require any `gluing'.
However, instead of gluing vertices together, we can initialize a pair as a Bell state, provided that the qubits have not been used so far.
In other words, we want to choose as large a matching as possible on our graph.
Suppose that the largest number of vertices we can match is $A$ vertices out of $V+n$ total, where $V=r_t n$.
Those vertices will have at most $2A$ remaining free edge ends.
It must be the case that all of the internal edges of the remaining vertices connect to our $A$ matched vertices, since if any did not, we would be able to add it to our matching.
The remaining $V+n-A$ vertices have $3V+3n-3A$ ends, of which at most $n$ are outputs, leaving $3V+2n-3A$ ends which must connect to our $2A$ free ends.
This gives us the bound $2A\geq3V+2n-3A$ so $A\geq\frac{3V+2n}5$.

Our graph has $3V+3n$ edge ends of which $3V+2n$ are internal, so it has $\frac{3V}2+n$ internal edges that must be linked up.
Normally, we would require $2V+2n$ CNOTS to make the 3-qubit \CAT states and another $\frac{3V}2+n$ to link the edges for $\frac{7V}2+3n$ total.
We would also require $3V+3n$ qubits total, of which $3V+2n$ would be ancillae.
However, every matched pair of our $A$ vertices saves us 2 CNOTs and 2 ancilla qubits, which yields a grand total of $\frac{7V}2+3n-\frac{3V+2n}{5}=\frac{29V+26n}{10}=\left(\frac{29r_t+26}{10}\right)n$ CNOTs and $3V+2n-\frac{3V+2n}{5}=\frac{12V+8n}{5}=\left(\frac{12r_t+8}5\right)n$ ancilla qubits.

Notice that we only need to go to these extremes to achieve the minimal depth of 3.
There is a very substantial tradeoff between depth and ancillae and CNOTs.
\end{proof}

\section{Conjecture on finding vertex ratios for larger t}\label{app:vertex-ratio-conjecture}

As we increase $t$, we can consider various optimal strategies for constructing a compact graph with many marks. If we consider subgraphs that contain no marks and for which every outgoing edge contains a mark, we find that making such subgraphs trees is optimal for the vertex ratio.
For the purposes of this section, we will imagine attaching to these trees half of each of their outgoing edges, and putting half a mark on each such half-edge.
This will allow us to glue trees by connecting two half-edges and letting their two half-marks merge into a single mark.
Let $k$-trees be such trees with $k$ half-edges and $k$ half-marks. Since our graph is always 3-regular, the tree itself will always contain $k-2$ vertices and one fewer full edges between them.
If we were to connect an $A$-tree to a $B$-tree, we find that cutting just outside all of the marks of our trees will cut at most $A+B-2$ edges including $A+B-1$ marks, which means this only is acceptable when $t\leq A+B-3$ or $t+3\leq A+B$.
In practice, we find that allowing only values of $A$ and $B$ that are as close to $\frac{t+3}2$ as possible yields graphs with the best vertex ratio.
Informally, this is because if we had a small value of $A$ and a large value of $B$, then when we connected an $A$-tree to a $B$-tree, we would need to connect trees larger than $A$ to the leaves of the $B$-tree close to where the first connection was made.

This evidence suggests that the best vertex ratio is obtained by connecting $A$-trees and $B$-trees in a $B:A$ ratio, while setting $A=\left\lfloor\frac{t+3}2\right\rfloor$ and $B=\left\lceil\frac{t+3}2\right\rceil$.
This yields a conjectured optimal vertex ratio of
\begin{equation}
\frac{2(AB-B-A)}{AB}=2-\frac{2A+2B}{AB}=2-\frac{2(t+3)}{\left\lfloor\frac{t+3}2\right\rfloor\left\lceil\frac{t+3}2\right\rceil}.
\end{equation}
This quantity is equal to $2-\frac{8}{t+3}$ when $t$ is odd and $2-\frac{8(t+3)}{(t+2)(t+4)}$ when $t$ is even.
From this formula, we also conjecture that $r_+\coloneqq\sup_t r_t=2$, which is not the same quantity as $r_\infty$. The quantity $r_+$ is a vertex ratio upper bound for arbitrary large but finite values of $t$, whereas $r_\infty$ is such an upper bound for infinite $t$, effectively limiting the size of any graph that can be constructed for a given $t$.

We see that the example constructions of optimal graph families shown above exactly match this pattern.
When constructing such patterns for larger $t$, it becomes more difficult to construct $t$-robust graphs and not accidentally admit any nonlocal cuts.
We conjecture that one solution to the tiling of such $A$- and $B$-trees involves using a tiling of a hyperbolic surface with $A$-gons and $B$-gons, and assigning the gluing of the trees in accordance with the gluing of the respective edges.
As long as the resulting graph is $t$-robust and uses the correct ratio of the two tree types, we conjecture it will have an optimal vertex ratio.

\section{Details on constructing optimal \CAT states in practice}\label{app:cat-state-gen}

We give here more details on the procedures described in \cref{sec:cat-state-generation} regarding how to construct optimal marked graphs in practice.

\subsection{Finding nonlocal cuts}\label{subsec:nonlocal-cut-SAT}

Let $G = (V, E)$ be a 3-regular graph.
We formulate the search for a nonlocal $t$-cut as a SAT problem.
To do so, we introduce the following sets of boolean variables:
\begin{description}
    \item[Partition variables:] $x_u$ for all $u \in V$: True if vertex $u$ is in partition $A$, and False if $u$ is in partition $B$.
    \item[Witness variables:] $a_u$ and $b_u$ for all $u \in V$: True if $u$ is part of the witness subset of partition $A$ or partition $B$, respectively.
    \item[Cut edge variables:] $d_{uv}$ for all $(u, v) \in E$: True if the edge crosses the cut.
\end{description}

\subsubsection{Cut Size Constraints}
To measure the size of the cut, we identify edges whose endpoints belong to different partitions.
For each edge $(u, v) \in E$, we enforce the implication $(x_u \neq x_v) \implies d_{uv}$ using:
\[
  (\neg x_u \lor x_v \lor d_{uv}) \land (x_u \lor \neg x_v \lor d_{uv})
\]
We then bound the total number of cut edges by applying a cardinality constraint over all $d_{uv}$ variables:
\[\sum_{(u, v) \in E} d_{uv} \leq t\]
\textit{(Note: In practice, this is translated into CNF using standard cardinality encoding methods provided by PySAT~\cite{imms-sat18, itk-sat24}).}

\subsubsection{Witness Subset Constraints for Partition A}
We require partition $A$ to contain a non-empty subset of vertices where each vertex has at least two neighbors in the same subset.
First, we ensure the subset is non-empty:
\[\bigvee_{u \in V} a_u\]
Next, we ensure that any vertex in the witness subset actually belongs to partition $A$:
\[\neg a_u \lor x_u \quad \forall u \in V\]
Finally, we enforce the degree constraint.
Because $G$ is 3-regular, every vertex $u$ has exactly three neighbors, say $n_1, n_2, n_3$.
If $u$ is in the witness set ($a_u$ is True), at most one of its neighbors can be excluded from the set.
This is equivalent to stating that any pair of $u$'s neighbors must contain at least one vertex also in the witness set.
We add the following clauses for every $u \in V$:
\begin{gather*}
    \neg a_u \lor a_{n_1} \lor a_{n_2}\\
    \neg a_u \lor a_{n_1} \lor a_{n_3}\\
    \neg a_u \lor a_{n_2} \lor a_{n_3}\\
\end{gather*}

We apply mirror constraints to ensure partition $B$ also contains a valid, non-empty witness subset.

\subsubsection{Symmetry Breaking}
To halve the search space, we fix the partition assignment of an arbitrary starting vertex $v_0 \in V$ to partition $A$ with a unit clause:
\[x_{v_0}\]

\subsection{Finding markings}\label{subsec:markings-SAT}

For checking the necessary conditions for a set of markings $M$ to be valid, we consider the structures imposed by unit trees (\cref{def:unit-cells}) for a given $t$, and in particular, those given for small values of $t$ in \cref{fig:unit-branch}.
We define a binary variable \(x_e \in \{0, 1\}\) for each edge \(e \in E\), where \(x_e = 1\) implies \(e \in M\).
The objective is to maximise the total weight:
\[
    \text{Maximise}\ \ \ \sum_{e \in E} x_e.
\]
This serves as the set of \textit{soft clauses} in our Weighted Conjunctive Normal Form (WCNF) instance.
The \textit{hard clauses} are derived from the local constraints for unit trees.

\paragraph*{Easy constraints for Fixed $t$.}

For specific small values of $t$, there is a straightforward encoding of the problem as standard CNF clauses.
Let \(G = (V, E)\) be the graph on which we are looking for markings, and \(L(G) = (V_L, E_L)\) the line graph of \(G\).
\begin{description}
    \item[For \(T = 3\):] All edges must be marked
    \item[For \(T = 4\):] Every marked edge must have an unmarked neighbor, which is encoded as the following CNF:
    \[
        \bigwedge_{e \in V_L\!} \left( \neg x_e \lor \bigvee_{f \in N(e)} \neg x_f \right)
    \]
    \item[For \(T = 5\):] Every vertex must have an unmarked neighbor, which is encoded as the following CNF:
    \[
        \bigwedge_{v \in V} \ \bigvee_{\substack{e \in E,\\ v \in e}} \neg x_e
    \]
    \item[For \(T = 7\):] Every edge must have an unmarked neighbor, which is encoded as the following CNF:
    \[
        \bigwedge_{e \in V_L\!} \ \bigvee_{f \in N(e)} \neg x_f
    \]
    Note that this constraint forces the unmarked edges to form a total dominating set~\cite{west2001introduction} for the line graph.
    Deciding if the number of vertices in the smallest dominating set is less then a fixed \(K\) is a classical NP-complete decision problem~\cite{garey1979computers}.
\end{description}

\paragraph*{General constraints for arbitrary $t$}
While there exist nice CNF encodings for given values of $t$, it is not possible to express such marking constraints for arbitrary values of $t$.
To determine the general constraint for a given $t$, we enumerates all subtrees \(S \subset V\) up to size \(|S| \le t-2\).
For each subtree, let \(E(S)^+\) denote the set of edges incident to or contained within \(S\).
We add a cardinality constraint to the SAT solver:
\[
    \sum_{e \in E(S)^+} x_e \le t.
\]
This ensures that no small local cluster contains more than $t$ marks.

\begin{algorithm}[h]
\caption{Heuristic construction of graphs without nonlocal T-cut}
\label{alg:construction}
\begin{algorithmic}[1]
\Require Number of vertices \(n\), Target $t$, Max iterations \(K\)
\Ensure A 3-regular graph \(G\) with no nonlocal $t$-cuts (or \textsc{Failure})

\State \textbf{Initialise:} \(G \gets\) Random 3-regular graph of size \(n\)
\State \textbf{Hyperparameters:} \(\lambda_{\text{thresh}} \gets \frac{10}{3} \cdot \frac{t}{n}\)

\For{\(k = 1\) to \(K\)}
    \State \(g \gets \text{Girth}(G)\)
    \State \(\lambda_2 \gets \text{AlgebraicConnectivity}(G)\)

    \If{\(g > t\) \textbf{and} \(\lambda_2 \ge \lambda_{\text{thresh}}\)}
        \If{\(\textsc{HasNonlocalCut}(G, t) = \text{False}\)}
            \State \Return \(G\)
        \EndIf
    \EndIf

    \State \(G' \gets \text{DoubleEdgeSwap}(G)\)
    \State \(g' \gets \text{Girth}(G')\)
    
    \If{\(g \le t\)} \Comment{\textbf{Phase 1: Girth Optimisation}}
        \If{\(g' \ge g\)}
            \State \(G \gets G'\)
        \EndIf
    \Else \Comment{\textbf{Phase 2: Spectral Optimisation}}
        \If{\(g' > t\)}
            \State \(\lambda_2' \gets \text{AlgebraicConnectivity}(G')\)
            \If{\(\lambda_2' > \lambda_2\)}
                \State \(G \gets G'\)
            \EndIf
        \EndIf
    \EndIf
\EndFor

\State \Return \textsc{Failure}
\end{algorithmic}
\end{algorithm}

\end{document}

%% file: externalize-disable.tex
\tikzexternaldisable  

%% file: figs/prelim/cup.tikz
\begin{tikzpicture}
	\begin{pgfonlayer}{nodelayer}
		\node [style=none] (12) at (0, -0.75) {};
		\node [style=none] (14) at (0, 0.75) {};
		\node [style=none] (15) at (0, 0.75) {};
		\node [style=none] (16) at (0, -0.75) {};
	\end{pgfonlayer}
	\begin{pgfonlayer}{edgelayer}
		\draw [bend left=90, looseness=1.75] (16.center) to (15.center);
	\end{pgfonlayer}
\end{tikzpicture}

%% file: figs/prelim/cap.tikz
\begin{tikzpicture}
	\begin{pgfonlayer}{nodelayer}
		\node [style=none] (12) at (0, -0.75) {};
		\node [style=none] (14) at (0, 0.75) {};
		\node [style=none] (15) at (0, 0.75) {};
		\node [style=none] (16) at (0, -0.75) {};
	\end{pgfonlayer}
	\begin{pgfonlayer}{edgelayer}
		\draw [bend right=90, looseness=1.75] (16.center) to (15.center);
	\end{pgfonlayer}
\end{tikzpicture}